\documentclass[a4paper,12pt]{article}

\makeatletter
 
  \@addtoreset{equation}{section}
 \makeatother
\usepackage{amsfonts}
\usepackage{amssymb}
\usepackage{mathrsfs}
\usepackage{amsmath,amsthm}
\usepackage{amsmath}

\usepackage{graphicx}
\usepackage{color}
\usepackage{indentfirst}
\usepackage{url}

\begin{document}

\newtheorem{thm}{Theorem}
\newtheorem{lem}{Lemma}
\newtheorem{prop}{Proposition}
\newtheorem{cor}{Corollary}
\theoremstyle{definition}
\newtheorem{defn}{Definition}
\newtheorem{remark}{Remark}
\newtheorem{step}{Step}

\newcommand{\Cov}{\mathop {\rm Cov}}
\newcommand{\Var}{\mathop {\rm Var}}
\newcommand{\E}{\mathop {\rm E}}
\newcommand{\const }{\mathop {\rm const }}
\everymath {\displaystyle}

\newcommand{\ruby}[2]{
\leavevmode
\setbox0=\hbox{#1}
\setbox1=\hbox{\tiny #2}
\ifdim\wd0>\wd1 \dimen0=\wd0 \else \dimen0=\wd1 \fi
\hbox{
\kanjiskip=0pt plus 2fil
\xkanjiskip=0pt plus 2fil
\vbox{
\hbox to \dimen0{
\small \hfil#2\hfil}
\nointerlineskip
\hbox to \dimen0{\mathstrut\hfil#1\hfil}}}}

\def\qedsymbol{$\blacksquare$}
\renewcommand{\thefootnote }{\fnsymbol{footnote}}
\renewcommand{\refname }{References}

\everymath {\displaystyle}

\title{
An Optimal Execution Problem with Market Impact
}
\author{Takashi Kato
\footnote{Division of Mathematical Science for Social Systems, 
              Graduate School of Engineering Science, 
              Osaka University, 
              1-3, Machikaneyama-cho, Toyonaka, Osaka 560-8531, Japan, 
E-mail: \texttt{kato@sigmath.es.osaka-u.ac.jp}}
}

\date{}
\maketitle 

\begin{abstract}
We study an optimal execution problem in a continuous-time market model that considers market impact. 
We formulate the problem as a stochastic control problem and 
investigate properties of the corresponding value function. 
We find that right-continuity at the time origin is associated with the strength of market impact for large sales, otherwise the value function is continuous. 
Moreover, we show the semi-group property (Bellman principle) 
and characterise the value function as a viscosity solution of 
the corresponding Hamilton--Jacobi--Bellman equation. 
We introduce some examples where 
the forms of the optimal strategies change completely, depending on the amount of 
the trader's security holdings 
and where optimal strategies in the Black--Scholes type market with nonlinear market impact are 
not block liquidation but gradual liquidation, even when the trader is risk-neutral. 

\footnote[0]{Mathematical Subject Classification (2010) \  91G80, 93E20, 49L20}\\
\footnote[0]{JEL Classification (2010) \ G33, G11}
\end{abstract}

\section{Introduction}\label{intro} 

An optimal portfolio management problem was developed in \cite {Merton1}, 
\cite {Merton2} and in other papers. 
These classical financial theories assumed that assets in the market are perfectly liquid, but real markets pose various liquidity risks. 
For instance, the problem of transaction costs and the uncertainty of trading. 

Another important problem of liquidity is market impact (MI), that is, the effect of 
trader investment behaviour on security prices. 
Such problems are often discussed in the framework of optimal execution problems, 
where a trader holds a certain amount of a security 
and tries to execute trades within a time horizon. 
The optimal execution problem considering MI was first studied in \cite {Bertsimas-Lo} 
as a minimisation problem of an expected execution cost in a discrete-time model, and 
that model was generalised to a mean-variance model 
in \cite {Almgren-Chriss} and \cite {Huberman-Stanzl2}. 
A continuous-time model of the execution problem was studied in 
\cite{He-Mamaysky}, \cite {Subramanian-Jarrow} and 
\cite{Vath_et_al} as a singular/impulse stochastic control problem. 
Forsyth \cite {Forsyth} has also studied the continuous-time model in the framework of mean-variance analyses and 
gave a viscosity characterisation of the corresponding value functions. 
An infinite time horizon case is treated in \cite {Schied-Schoeneborn}. 
The optimal execution problem in the limit-order-book (LOB) model is studied 
in \cite {Alfonsi-Fruth-Schied}, \cite {Alfonsi-Schied}, 
\cite {Alfonsi-Schied-Slynko}, \cite {Gatheral}, 
\cite{Gatheral-et-al}, \cite {Predoiu-et-al}, etc. 

Recently, various studies have examined the optimisation problem with MI, 
but a standard framework has yet to be established. 
In this paper we develop a mathematical model of optimal execution. 
Our model is formulated as a stochastic control problem in the continuous-time model, 
which is characterised as the limit of those of the discrete-time models 
(see Section \ref {sec_discrete}). 

We study our optimal execution model 
by investigating the properties of the corresponding value function. 
First we study the continuity of the value function. 
We find that our value function is continuous in each parameter except for the time origin, where 
the right-continuity at $t = 0$ is according to the ``strength'' of MI function. 
This implies that instantaneous liquidation of large volume 
makes no sense when MI for large trade is strong 
(in a meaning to be discussed later). 
Next we show the semi-group property (the Bellman principle) of the value function. 
This property is standard in the theory of stochastic control, but a strict proof is generally difficult. 
We show the semi-group property of our value function by applying an argument based on Nisio's method 
(\cite {Nisio}). 

The semi-group property suggests that the value function is characterised as a viscosity solution of 
the corresponding Hamilton--Jacobi--Bellman equation (HJB), 
which is a nonlinear second order partial differential equation (PDE). 
We prove that our value function actually becomes a viscosity solution of the HJB, 
and we also show the uniqueness of viscosity solutions of the HJB under certain mathematical assumptions. 
Note that the viscosity characterisation of value functions of stochastic control problems 
has been broadly studied elsewhere (\cite {Fleming-Soner}, \cite{Nagai} and 
the references therein). 
Uniqueness (or the comparison principle) of viscosity solutions of HJB is also well-studied 
in the theory of nonlinear PDEs (see, for instance, \cite {DaLio-Ley1}, \cite {DaLio-Ley2}). 
Nevertheless, our results are original: 
we cannot apply the existing results to our model 
because of the unboundedness of the control region and the growth conditions of the coefficients 
(see Remark 4 for details). 

We mainly consider the case where the MI function is convex with respect to execution volume. 
Although some empirical studies tell us that the MI function becomes concave according to market circumstances 
(see, for instance,  \cite {Almgren-et-al}), 
considering the effect of a convex MI is interesting and important from a theoretical viewpoint. 
We give some examples 
that imply that a convex MI function makes a trader avoid a block liquidation 
(selling all of a security at once) and induce a gradual liquidation 
(selling over a period of time). 
As another interesting result, 
we find that the forms of a risk-neutral trader in the Black--Scholes type model 
with a (log-)quadratic MI function drastically change according to the initial shares of the security held. 

This paper is organised as follows: 
In Section \ref {section_Model}, we introduce our model. 
In Section \ref {section_results}, we give our main results. 
We study properties of our value function, namely 
continuity, the semi-group property and 
the viscosity characterisation. 
Moreover, we have the uniqueness result of the viscosity solution of the corresponding HJB when MI is sufficiently strong. 
In Section \ref {sec_SO}, 
we also consider the case where the trader needs to liquidate all holdings of the security. 
We show that such a `sell-off' (liquidation) condition does not influence the form of the value function in our model. 
In Section \ref {sec_eg}, we consider some examples of our model. 
We summarise this paper in Section \ref {sec_summary}. 
In Section \ref {sec_discrete}, we introduce the derivation of our continuous-time model 
from the discrete-time models, and we give proofs of our results in Section \ref {Proofs}. 

\section{The Model}\label{section_Model}
\setcounter{thm}{0}
\setcounter{equation}{0}

We now present the details of the model. 
Let $(\Omega ,\mathcal {F}, (\mathcal {F}_t)_{0\leq t\leq T}, P)$ 
be a filtered space that satisfies the usual conditions (i.e., 
$(\mathcal {F}_t)_t$ is right-continuous and $\mathcal {F}_0$ contains 
all $P$-null sets), and let 
$(B_t)_{0\leq t\leq T}$ be a standard one-dimensional 
$(\mathcal {F}_t)_t$-Brownian motion. 
Here, $T > 0$ denotes the time horizon. 
For simplicity we assume $T = 1$. 

Suppose that the market consists of one risk-free asset (cash) and 
one risky asset (a security). 
The price of cash is always $1$, which means that 
the risk-free rate is zero. 
The price of the security fluctuates according to a certain stochastic flow, 
and is influenced by the trader's sales. 
We consider a single trader who has an endowment of $\Phi _0>0$ shares of a security. 
This trader liquidates the shares $\Phi _0$ over a time interval $[0,1]$, 
but these sales affect the price of the security. 

First, we define trading strategies. 
We say that a stochastic process $(\zeta _r)_{0\leq r\leq 1}$ is 
an admissible execution strategy if $(\zeta _r)_r\in \mathcal {A}_1(\Phi _0)$, where 
\begin{eqnarray}\nonumber 
\mathcal {A}_t(\varphi ) &=& \Big\{ (\zeta _r)_{0\leq r\leq t}\ ; \ 
(\mathcal {F}_r)_{r}\mbox {-progressively measurable}, \ \ 
\mbox{nonnegative}, \\&&\hspace{20mm}
\int ^t_0\zeta _rdr\leq \varphi \ \mbox{a.s.}, \ \ 
\sup _{r,\omega }\zeta _r(\omega )<\infty  
\Big\} . 
\end{eqnarray}
The value $\zeta _r$ represents the instantaneous sales (i.e., execution speed) at time $r$. 
To avoid technical difficulties, we do not consider short selling in this paper. 
Thus $\zeta _r$ is assumed to be nonnegative. 
The integral $\int ^t_0\zeta _rdr$ denotes the cumulative volume of liquidation until time $t$. 
The trader cannot liquidate more of the security than the initial shares held. 
The inequality $\int ^1_0\zeta _rdr\leq \Phi _0$ represents such a situation. 
The boundedness $\sup _{r, \omega }\zeta _r(\omega ) < \infty $ has no financial meaning; 
the arguments in Section \ref {sec_discrete} 
imply that we may consider the continuous-time 
optimal execution problem in only the case where trading strategies satisfy this condition 
(see also Remark \ref {rem_discrete} 
below). 

Next we introduce security price fluctuations
and define our MI function. 
Let $s_0 > 0$ be the initial price and let $x_0 = \log s_0$. 
$S_t$ is the security price at time $t$ and $X_t$ is its log-price, so 
$X_t = \log S_t$. 
If the trader does not trade, meaning there is no MI, 
the fluctuation of $(X_t)_t$ is described by the following stochastic differential equation (SDE): 
\begin{eqnarray}\label{SDE_X0}
\left\{
\begin{array}{ll}
 	dX_r = \sigma (X_r)dB_r+b(X_r)dr,&	\\
 	\hspace{2mm}X_0= x_0, &
\end{array}
\right.
\end{eqnarray}
where $b, \sigma  : \Bbb {R}\longrightarrow \Bbb {R}$ are Borel functions. 
We assume that $b$ and $\sigma $ are bounded and Lipschitz continuous, 
so there exists a unique solution of (\ref {SDE_X0}). 
When the trader liquidates the security with the liquidation schedule $(\zeta _r)_r\in \mathcal {A}_1(\Phi _0)$, 
the log-price is influenced by MI. 
In this case $(X_r)_r$ follows
\begin{eqnarray}\label{SDE_X}
\left\{
\begin{array}{ll}
 	dX_r = \sigma (X_r)dB_r+b(X_r)dr - g(\zeta _r)dr,& 	\\
 	\hspace{2mm}X_0 = x_0. &
\end{array}
\right.
\end{eqnarray}
Here, $g : [0, \infty )\longrightarrow [0, \infty )$ is our (permanent) MI function, and 
at time $t$ the log-price decreases by $g(\zeta _t)dt$ according to the 
execution speed $\zeta _t$. 
We assume that $g$ is non-decreasing and continuously differentiable; let 
\begin{eqnarray}
h(\zeta ) = g'(\zeta )
\end{eqnarray}
be its derivative. 
Then $h$ represents how MI becomes large when the trading size increases. 
In this paper we consider the case where MI function $g$ is convex 
(not necessarily strictly convex; $g$ may be assumed to be a linear function); 
thus $h$ is a non-decreasing function. 
This assumption is required to show Theorem \ref {converge} below. 
Moreover, the convexity of $g$ plays an essential role in our execution problem, 
in that this property motivates the trader to liquidate the shares of the security 
spending time to avoid a huge MI (see Sections \ref {sec_eg}--\ref {sec_summary}). 

Note that the process $(S_r)_r$ satisfies 
\begin{eqnarray}\label{SDE_S}
\left\{
\begin{array}{ll}
dS_r = \hat{\sigma }(S_r)dB_r+\hat{b}(S_r)dr-g(\zeta  _r)S_rdr, & \\
 	\hspace{2mm}S_0= s_0, &
\end{array}
\right.
\end{eqnarray}
where $\hat{\sigma }(s) = s\sigma (\log s)$ and $\hat{b}(s) = s\{ 
b(\log s)+\sigma (\log s)^2/2\} $. 
The existence and uniqueness of solutions of SDEs (\ref {SDE_X})--(\ref {SDE_S}) also hold 
for each admissible strategy $(\zeta _r)_r$. 

The trader's problem is to choose an admissible strategy to 
maximise the expected utility $\E [u(W_1, \varphi _1, S_1)]$ 
for the utility function $u\in \mathcal {C}$, 
where $W_t$ (resp. $\varphi _t$) represents the amount of cash (resp. security) 
at time $t$, and 
$\mathcal {C}$ denotes the set of non-decreasing continuous functions 
on $D = \Bbb {R}\times [0,\Phi _0]\times [0,\infty )$ such that 
\begin{eqnarray}\label{growth_C}
u(w,\varphi ,s)\leq C_u(1+w^2+s^2)^{m_u}, \ \ (w,\varphi ,s)\in D
\end{eqnarray}
for some constants $C_u>0$ and $m_u\in \Bbb {N}$ (i.e., $u$ has polynomial growth rate). 
Mathematically, this problem is characterised by the value function 
\begin{eqnarray}\nonumber 
V_t(w,\varphi ,s ; u) = \sup _{(\zeta _r)_{r}\in \mathcal {A}_t(\varphi )}
\E [u(W_t,\varphi _t, S_t)] 
\end{eqnarray}
subject to 
\begin{eqnarray}\label{init_cond}
(W_0, \varphi _0, S_0) = (w, \varphi , s) 
\end{eqnarray}
and 
\begin{eqnarray}\nonumber 
dW_r &=& \zeta _rS_rdr, \\\label{value_conti_S}
d\varphi _r &=& -\zeta _rdr, \\\nonumber 
dS_r &=& \hat{\sigma }(S_r)dB_r+\hat{b}(S_r)dr-g(\zeta  _r)S_rdr 
\end{eqnarray}
for $t\in [0,1], (w,\varphi ,s)\in D$, and $u\in \mathcal {C}$. 
We remark that $V_0(w,\varphi ,s ; u) = u(w,\varphi ,s)$. 
Also note that $V_t(w,\varphi ,s ; u)<\infty $ for any $t\in [0,1]$, 
$(w,\varphi ,s)\in D$ and $u\in \mathcal {C}$ (see Sections \ref {pre}--\ref {S-R}). 

For technical reasons, we allow the security price to take the value $0$ 
(note that $\hat{b}(0)$ and $\hat{\sigma }(0)$ are defined as $0$). 
When $s > 0$, we can obviously rewrite 
\begin{eqnarray*}
V_t(w,\varphi ,s ; u) = \sup _{(\zeta _r)_{r}\in \mathcal {A}_t(\varphi )}
\E [u(W_t,\varphi _t, S_t)] 
\end{eqnarray*}
subject to (\ref {init_cond}) and 
\begin{eqnarray}\nonumber 
dW_r &=& \zeta _r\exp (X_r)dr,\\\label{value_conti_X}
d\varphi _r &=& -\zeta _rdr, \\\nonumber 
dX_r &=& \sigma (X_r)dB_r+b(X_r)dr-g(\zeta  _r)dr , \\\nonumber 
S_r &=& \exp (X_r). 
\end{eqnarray}
For convenience, we denote a triplet 
$(W_r, \varphi _r, S_r)_{0\leq r\leq t}$ of (\ref {init_cond}), (\ref {value_conti_S}) 
by $\Xi _t(w,\varphi ,s ; \allowbreak (\zeta _r)_r)$, 
and 
$(W_r, \varphi _r, \allowbreak X_r)_{0\leq r\leq t}$ of (\ref {init_cond}), (\ref {value_conti_X}) by 
$\Xi ^X_t(w,\varphi ,s ; \allowbreak (\zeta _r)_r)$. 

Note that a trader whose execution strategy is in $\mathcal {A}_t(\varphi )$ is permitted 
to leave shares of the security unsold, 
and there will be no penalty if the trader cannot finish the liquidation within the time horizon. 
In Section \ref {sec_SO}, we consider a case when the trader must finish the liquidation.

\begin{remark}\label{rem_discrete}
The definition of $V_t(w, \varphi , s ; u)$ originates from 
a convergence theorem for value functions of optimal execution problems from discrete-time to continuous-time models. 
As pointed out in \cite {Kato_FMA}, 
in constructing a mathematical model of a financial problem, 
the discrete-time model, on the one hand, significantly describes realistic phenomena exactly, 
but sometimes it is hard to get a clean model due to complex noise. 
The continuous-time model, on the other hand, often makes problems clear, 
but the superficial construction of continuous-time models may overlook the essence of the problem. 
Therefore, 
it is meaningful to construct an adequate model 
by the following procedures: we first considered a discrete-time model of an optimal execution problem with MI, 
and then derived a continuous-time model as the limit. 
In fact, our value function $V_t(w, \varphi , s ; u)$ 
is derived in such a way from given discrete-time models. 
Please refer to Section \ref {sec_discrete} for details. 
\end{remark}

\begin{remark}\label{rem_temporary}
MI can be divided into two parts: a permanent impact and a temporary impact 
(see \cite{Almgren-Chriss} and \cite {Holthausen-et-al}). 
As time passes, the temporary impact disappears and the transitorily depressed price recovers. 
Our MI function $g(\zeta )$ corresponds to the permanent impact. 

We can define a value function of the optimal execution problem 
with both permanent and temporary MI in a continuous-time model such as  
\begin{eqnarray*}
\hat{V}_t(w,\varphi ,s ; u) &=& \sup _{(\zeta _r)_{r}\in \mathcal {A}_t(\varphi )}
\E [u(W_t,\varphi _t, S_t)] \\
&\mathrm {s.t.}&\hspace{0.1mm}
dW_r = \zeta _r\exp (X_r - \tilde{g}(\zeta _r))dr,\\
&&\hspace{1.3mm}d\varphi _r=-\zeta _rdr, \\
&&\hspace{0.8mm}dX_r = \sigma (X_r)dB_r+b(X_r)dr-g(\zeta  _r)dr , \\
&&\hspace{3.2mm}S_r = \exp (X_r), \\
&&(W_0, \varphi _0, S_0) = (w, \varphi , s), 
\end{eqnarray*}
where $\tilde{g}(\zeta )$ denotes the temporary MI function. 
We can also show continuity of this value function in $w, \varphi $ and $s$. 
However, constructing a discrete-time version of the problem is difficult 
due to technical reasons. 
Moreover, the Bellman principle (Theorem \ref {semi} in the next section) is proved by 
Nisio's method, which is based on a discrete-time approximation of the value function. 
Since Theorem \ref {semi} plays an essential role in proving Theorem \ref {conti} (especially continuity in $t$) 
and Theorems \ref {HJB_th} and \ref {unique_th}, 
we cannot sufficiently study the properties of our value function 
when there is temporary MI. 
In this paper, therefore, we treat only permanent MI functions. 
For further comments, see Section \ref {sec_summary}. 
\end{remark}

\section{Main Results}\label{section_results}
\setcounter{thm}{0}

Next we present the main results of this paper. 
First we introduce the result of the continuity of $V_t(w,\varphi ,s ; u)$. 
Here we denote $h(\infty ) = \lim _{\zeta \rightarrow \infty }h(\zeta )$ for brevity. 

\begin{thm}\label{conti} \ Let $u\in \mathcal {C}$.\\
$\mathrm {(i)}$ \ If $h(\infty )=\infty $, then
$V_t(w,\varphi ,s ; u)$ is continuous in $(t,w,\varphi ,s)\in [0,1]\times D$. \\
$\mathrm {(ii)}$ \ If $h(\infty )<\infty $, then 
$V_t(w,\varphi ,s ; u)$ is continuous in $(t,w,\varphi ,s)\in (0,1]\times D$ and 
$V_t(w,\varphi ,s ; u)$ converges to $Ju(w,\varphi ,s)$ 
uniformly on any compact subset of $D$ as $t\downarrow 0$, where 
\begin{eqnarray*}
Ju(w,\varphi ,s) = \left\{
                   \begin{array}{ll}
                    	\sup _{\psi \in [0,\varphi ]}
u\Big (w+\frac{1-e^{-h(\infty )\psi }}{h(\infty )}s,\varphi -\psi ,se^{-h(\infty )\psi} \Big )& 
 (h(\infty )>0)	\\
                    	\sup _{\psi \in [0,\varphi ]}
u(w+\psi s,\varphi -\psi ,s)& (h(\infty )=0). 
                   \end{array}
                   \right. 
\end{eqnarray*}
\end{thm}

We prove Theorem \ref {conti} in Section \ref {proof_of_conti}. 
As we can see, continuity in $t$ at the origin depends on
the state of the function $h$ at infinity. 
When $h(\infty ) = \infty $, MI of large sales 
is strong enough ($g(\zeta )$ diverges rapidly with $\zeta \rightarrow \infty $) to make a trader avoid instantaneous liquidation.
An optimal policy is `no-trading' in an infinitesimal time; thus, $V_t$ converges to $u$ as $t\downarrow 0$. 
When $h(\infty ) < \infty $, the value function is not always continuous at $t = 0$ and has the right limit $Ju(w, \varphi ,s)$. 
In this case, MI for large sales 
is not as strong 
($g(\zeta )$ diverges, but the divergence speed is slow) 
and there is room for successful liquidation in the infinitesimal time. 
The function $Ju(w,\varphi ,s)$ corresponds to the utility of liquidation by a trader 
who sells a part of the shares of a security $\psi $ by dividing infinitely within an infinitely short time 
(such that the fluctuation of the security price is negligible) 
and who is left with the amount $\varphi -\psi $, that is, \begin{eqnarray}\label{almost_block}
\zeta ^\delta _r = \frac{\psi }{\delta }1_{[0, \delta ]}(r), \ \ r\in [0, t] \ \ (\delta \downarrow 0). 
\end{eqnarray}
Such a strategy is also discussed in \cite {Lions-Lasry}. 
We remark that the form of $Ju$ is strongly related to Theorem 3 in \cite {Lions-Lasry} 
(see Theorem \ref {th_LL} in Section \ref {sec_SO} for more details). 
Also note that the condition $h(\infty ) = 0$ corresponds to the classical case of no MI model. 

Next we study the semi-group property (the Bellman principle) of the family of 
nonlinear operators corresponding with the continuous-time value function. 
We define an operator $Q_t : \mathcal {C}\longrightarrow \mathcal {C}$ by 
$Q_tu(w,\varphi ,s) = V_t(w,\varphi ,s ; u)$. 
Here the arguments in Sections \ref {pre}--\ref {S-R} imply that $Q_t$ is well-defined. 
We now have the following theorem. 
\begin{thm}\label{semi} \ For each $r,t\in [0,1]$ with $t+r\leq 1$, \ 
$(w,\varphi ,s)\in D$ and $u\in \mathcal {C}$ 
it holds that $Q_{t+r}u(w,\varphi ,s)=Q_tQ_ru(w,\varphi ,s)$. 
\end{thm}

The proof is given in Section \ref {sec_properties}. 
Using Theorem \ref {semi}, 
we can characterise the continuous-time value function as 
the viscosity solution of the corresponding HJB. 
Since the value functions are defined in a way that does not depend on $\Phi _0$, 
we can take them to be defined on an extended domain 
$\hat{D} = \Bbb {R}\times [0, \infty ) \times [0, \infty )$. 
Let $u(w, \varphi , s) : \hat{D}\longrightarrow \Bbb {R}$ be such that 
$u$ is a non-decreasing continuous function that 
grows polynomially in $w, \varphi $ and $s$. 
We define a function $F : \mathscr {S}\longrightarrow [-\infty , \infty )$ by 
\begin{eqnarray*}
F(z, p, X) = -\sup _{\zeta \geq 0}
\left\{ \frac{1}{2}\hat{\sigma }(z_s)^2X_{ss} + \hat{b}(z_s)p_s + 
\zeta \left( z_sp_w - p_\varphi \right) - g(\zeta )z_sp_s\right\} , 
\end{eqnarray*}
where $\mathscr {S} = \hat{U}\times \Bbb {R}^3\times S^3$, $\hat{U} = \hat{D} \setminus \partial \hat{D}$, 
$S^3$ is the space of symmetric matrices in $\Bbb{R}^3\otimes \Bbb {R}^3$, and 
\begin{eqnarray*}
z = \left(
    \begin{array}{c}
     	z_w	\\
     	z_\varphi 	\\
     	z_s
    \end{array}
    \right)\in D, \ 
    p = \left(
    \begin{array}{c}
     	p_w	\\
     	p_\varphi 	\\
     	p_s
    \end{array}
    \right)\in \Bbb {R}^3, \ 
    X = \left(
    \begin{array}{ccc}
     	X_{ww}& X_{w\varphi }	& X_{ws}	\\
     	X_{\varphi w}& X_{\varphi \varphi }	&X_{\varphi s} 	\\
     	X_{sw}& X_{s\varphi }	& X_{ss}
    \end{array}
    \right) \in S^3. 
\end{eqnarray*}

Although the function $F$ approach $-\infty $, 
we can define a viscosity solution of the following HJB as usual 
(see, e.g., \cite{Fleming-Soner}, \cite{Koike} and \cite{Nagai}): 
\begin{eqnarray} \label{HJB2}
\frac{\partial }{\partial t}v + F(z, \mathcal {D}v, \mathcal {D}^2v) = 0 \ \ 
\mathrm {on} \ (0, 1]\times \hat{U}, 
\end{eqnarray}
where $\mathcal {D}$ denotes the differential operator with respect to 
$z = (w, \varphi , s)$. 
Here, we remark that $(\ref {HJB2})$ can be rewritten as 
\begin{eqnarray} \label{HJB}
\frac{\partial }{\partial t}v(t,w,\varphi ,s) - \sup _{\zeta \geq 0}
\mathscr {L}^\zeta v(t,w,\varphi ,s) = 0, \ \ (t, w, \varphi ,s)\in (0, 1]\times \hat{U}, 
\end{eqnarray}
where 
\begin{eqnarray*}
&&\mathscr {L}^\zeta v(t,w,\varphi ,s) = 
\frac{1}{2}\hat{\sigma }(s)^2\frac{\partial ^2}{\partial s^2}v(t,w,\varphi ,s) + 
\hat{b}(s)\frac{\partial }{\partial s}v(t,w,\varphi ,s)\\&&\hspace{20mm} + 
\zeta \Big (s\frac{\partial }{\partial w}v(t,w,\varphi ,s) - 
\frac{\partial }{\partial \varphi }v(t,w,\varphi ,s)\Big ) - 
g(\zeta )s\frac{\partial }{\partial s}v(t,w,\varphi ,s) . 
\end{eqnarray*} 

Now we state the following theorem, which is proved in Section \ref {sec_HJB}. 
\begin{thm} \ \label{HJB_th} Assume that $h$ is strictly increasing and $h(\infty ) = \infty $. 
Moreover, assume 
\begin{eqnarray}\label{cond_diff}
\liminf _{\varepsilon \downarrow 0}
\frac{V_t(w, \varphi ,s+\varepsilon  ; u) - V_t(w, \varphi , s ; u)}{\varepsilon } > 0
\end{eqnarray} 
for any $t\in (0, 1]$ and $(w, \varphi , s)\in \hat{U}$. 
Then $V_t(w,\varphi ,s ; u)$ is a viscosity solution of $(\ref {HJB2})$. 
\end{thm}

This theorem tells us that when $h(\infty ) = \infty $ (MI is strong), 
our value function is characterised by the corresponding HJB (\ref {HJB2}). 

\begin{remark}
It is quite natural that the value function is increasing with respect to the underlying security price, 
and we can easily prove that $V_t(w, \varphi , s ; u)$ is non-decreasing in $s$. 
Then it follows that 
\begin{eqnarray*}
\liminf _{\varepsilon \downarrow 0}
\frac{V_t(w, \varphi ,s+\varepsilon  ; u) - V_t(w, \varphi , s ; u)}{\varepsilon } \geq 0. 
\end{eqnarray*} 

The inequality (\ref {cond_diff}) is stricter than the one above, and is needed to prove our characterisation result (Theorem \ref {HJB_th}) 
because of technical reasons related to when $F = -\infty $ occurs. 
Note also that in many cases we may easily show that 
the value function is strictly increasing in $s$, that is, 
$V_t(w, \varphi ,s+\varepsilon  ; u) - V_t(w, \varphi , s ; u) > 0$ 
for $\varepsilon > 0$. 
Nevertheless, this does not directly indicate (\ref {cond_diff}). 
Here, we present a sufficient condition for (\ref {cond_diff}). 
\begin{itemize}
 \item [\mbox {[C1]}] \ $u(w, \varphi , s) = U(w)$ for some $U\in C^1(\Bbb {R})$. 
Moreover, $U$ is concave and $U'(w) \geq \delta $, \ $w\in \Bbb {R}$ for some $\delta > 0$. 
 \item [\mbox {[C2]}] \ Coefficients $b$ and $\sigma $ in (\ref {SDE_X}) 
are differentiable and their derivatives are 
Lipschitz continuous and uniformly bounded. 
\end{itemize}
Then we can show the following proposition:
\begin{prop} \ \label{prop_suff_cond}
Assume $[C1]$--$[C2]$. Then the inequality $(\ref {cond_diff})$ holds. 
\end{prop}
The proof is in Section \ref {proof_suff_cond}. 
Note that the conditions [C1]--[C2] are satisfied in typical cases. 
[C1] corresponds to the risk-averse (or risk-neutral) trader, 
which is standard in finance. 
[C1] further requires that the utility function depends only on the cash holdings $w$, 
but this assumption is also mild and standard 
(especially when we consider the sell-off condition, 
which will be discussed in Section \ref {sec_SO}). 
[C2] is satisfied in typical cases, such as the Black--Scholes model: 
In Section \ref {sec_eg}, 
we will treat such an example when the trader is risk-neutral. 
\end{remark}

Finally, we give the uniqueness result of viscosity solutions of (\ref {HJB}). 

\begin{thm} \ \label{unique_th} 
Assume that $\hat{\sigma }$ and $\hat{b}$ are both Lipschitz continuous. 
Assume the hypotheses of Theorem \ref {HJB_th} and that 
$\liminf _{\zeta \rightarrow \infty }(h(\zeta )/\zeta ) > 0$. 
If a polynomial growth function $v : [0, 1]\times \hat{D} \longrightarrow  \Bbb {R}$ is 
a viscosity solution of $(\ref {HJB})$ and satisfies the boundary conditions
\begin{eqnarray}\label{bdd_cond}
\begin{array}{rl}
v(0, w, \varphi , s) = u(w, \varphi , s), &(w, \varphi ,s)\in \hat{D}, \\
v(t, w, 0, s) = \mathrm {E}\left [u\left(w, 0, Z_t(s) \right)\right], & 
(t, w, s)\in [0, 1]\times \Bbb {R}\times [0, \infty ), \\
v(t, w, \varphi , 0) = u(w, \varphi , 0), & 
(t, w, \varphi )\in [0, 1]\times \Bbb {R}\times [0, \infty ), 
\end{array}
\end{eqnarray}
then $V_t(w, \varphi ,s ; u) = v(t, w, \varphi , s)$, where 
\begin{eqnarray}
Z_t(s) = 
\exp \left( Y_t(\log s) \right) \ (s > 0), \ \ 
0 \ (s = 0) 
\end{eqnarray}
and $Y_t(x)$ is the solution of SDE $(\ref {SDE_X0})$ replacing $x_0$ with $x$, 
that is, where $Z_t(s)$ represents the price of the security with no MI. 
\end{thm}

The proof of Theorem \ref {unique_th} is given in Section \ref {sec_uniqueness}. 
This theorem guarantees the uniqueness of viscosity solutions of (\ref {HJB}) 
when the divergence speed of $g(\zeta )$ with $\zeta \rightarrow \infty $ 
is greater than or equal to a quadratic function, that is, when 
$g(\zeta ) \geq C\zeta ^2$, $\zeta \geq M$ for some $C, M > 0$. 
In Section \ref {sec_quad_eg}, we present an example where 
the assumptions in Theorems \ref {HJB_th} and \ref {unique_th} are fulfilled. 

\begin{remark}
Characterisation of value functions of a stochastic control problem as viscosity solutions of HJB 
has been discussed in many papers and textbooks 
(e.g., \cite {DaLio-Ley1}, \cite {DaLio-Ley2}, \cite {Fleming-Soner} and \cite {Nagai}). 
Uniqueness results of viscosity solutions of HJB are also well studied. 
Yet to the best of our knowledge, the characterisation theorem (Theorem \ref {HJB_th}) and 
the uniqueness theorem (Theorem \ref {unique_th}) of our HJB (\ref {HJB2}) cannot be derived 
from the existing literature. 
The main difficulties are as follows: 
\begin{itemize}
 \item Our control region $[0, \infty )$ is unbounded. 
 \item The drift term $\hat{b}(s) - g(\zeta )s$ does not always satisfy the linear growth condition in $s$ and $\zeta $. 
In particular, if $h(\infty ) = \infty $, then we never get the estimates 
\begin{eqnarray*}
|g(\zeta )s| \leq C(1 + \zeta  + s), \ \ 
|g(\zeta )s - g(\zeta )s'| \leq C(1 + \zeta )|s-s'|, \ \ s, s', \zeta \geq 0
\end{eqnarray*}
for any positive constant $C$. 
\end{itemize}
Recently, the uniqueness of viscosity solutions of HJB has been studied for 
unbounded domains (in our case, $\hat{D} = \Bbb {R}\times [0, \infty )\times [0, \infty )$) and 
unbounded control regions (in our case, $[0, \infty )$). 
Theorem 2.1 of \cite {DaLio-Ley2} 
is one of the most general results of the comparison principle of viscosity solutions of HJB. 
However, our HJB does not satisfy conditions (A)(ii)--(iii) in \cite {DaLio-Ley2}. 
Thus, we cannot apply those results, meaning that our results are original in this respect.  
\end{remark}

\section{Sell-Off Condition}\label{sec_SO}
\setcounter{thm}{0}

In this section we consider the optimal execution problem under the `sell-off condition'. 
A trader has a certain quantity of shares of a security at some initial time, and 
must liquidate all of them within a time horizon. 
Then the spaces of admissible strategies are reduced to 
\begin{eqnarray*}
\mathcal {A}^{\mathrm {SO}}_t(\varphi ) = 
\left \{ (\zeta _r)_r\in \mathcal {A}_t(\varphi )\ ; \ 
\int ^t_0\zeta _rdr = \varphi \right \} . 
\end{eqnarray*}
We define a value function with the sell-off condition by 
\begin{eqnarray*}
V^{\mathrm {SO}}_t(w,\varphi ,s ; U) = \sup _{(\zeta _r)_r\in \mathcal {A}^{\mathrm {SO}}_t(\varphi )}\E [U(W_t)] 
\end{eqnarray*}
subject to (\ref {init_cond})--(\ref {value_conti_S}) 
for a continuous, non-decreasing, polynomially growing function $U : \Bbb {R}\longrightarrow \Bbb {R}$. 
This gives the following theorem: 

\begin{thm} \ \label{thesame}$V^{\mathrm {SO}}_t(w,\varphi ,s ; U) = V_t(w,\varphi ,s ; u)$, 
where $u(w, \varphi , s) = U(w)$. 
\end{thm}

\begin{proof}
The relation $V^{\mathrm {SO}}_t(w,\varphi ,s ; U) \leq  V_t(w,\varphi ,s ; u)$ is trivial, 
so we will show only the assertion $V^{\mathrm {SO}}_t(w,\varphi ,s ; U) \geq  V_t(w,\varphi ,s ; u)$. 
Take any $(\zeta _r)_r\in \mathcal {A}_t(\varphi )$ and let 
$(W_r, \varphi _r, S_r)_r = \Xi _1(w, \varphi ,s ; (\zeta _r)_r)$. 
Moreover take any $\delta \in (0, t)$. 
We define an execution strategy $(\zeta ^\delta _r)_r\in \mathcal {A}^{\mathrm {SO}}_t(\varphi )$ by 
$\zeta ^\delta _r = \zeta _r\ (r\in [0, t-\delta ]), \ \ 
\varphi _{t-\delta }/\delta  \ (r\in (t - \delta , t])$. 
Let $(W^\delta _r, \varphi ^\delta _r, S^\delta _r)_r = \Xi _1(w, \varphi ,s ; (\zeta ^\delta _r)_r)$. 
Then we have 
$W_{t-\delta } = W^\delta _{t-\delta } \leq  W^\delta _t$. 
Thus we get 
$\E [U(W_{t-\delta })]\leq \E [U(W^\delta _t)]\leq V^{\mathrm {SO}}_t(w, \varphi , s ; U)$. 
Letting $\delta \downarrow 0$, we have 
$\E [U(W_t)]\leq V^{\mathrm {SO}}_t(w, \varphi , s ; U)$ by the monotone convergence theorem 
(note that $U(W_{t-\delta })\geq U(w) > -\infty $). 
Since $(\zeta _r)_r\in \mathcal {A}_t(\varphi )$ is arbitrary, we obtain the assertion. 
\end{proof}

By Theorem \ref {thesame}, we see that the sell-off condition $\int ^t_0\zeta _rdr = \varphi $ makes no change 
in the (value of the) value function. 
No such phenomenon is observed in a discrete-time framework; 
although the value function in a discrete-time model in Section \ref {sec_discrete} 
may depend on whether the sell-off condition is imposed, 
in the continuous-time model we need not worry about such a condition.

When $g(\zeta )$ is linear, we can apply the variable reduction method ($9'$)--($12'$) in \cite {Lions-Lasry}\footnote{The author thanks Professor N. Touzi for pointing out this reference.} 
to obtain the following: 

\begin{thm} \ \label{th_LL}Assume $g(\zeta ) = \alpha \zeta $ for $\alpha > 0$. \\
$\mathrm {(i)}$ \ $V^\mathrm {SO}_t(w, \varphi , s ; U) = \overline{V}^\varphi _t
\left( w + \frac{1 - e^{-\alpha \varphi }}{\alpha }s, e^{-\alpha \varphi }s ; U\right)$, where 
\begin{eqnarray*}
\overline{V}^\varphi _t(\bar{w}, \bar{s} ; U) &=& \sup _{(\bar{\varphi }_r)_r\in \overline {\mathcal {A}}_t(\varphi )}
\E [U(\bar{W}_t)]\\
&&\hspace{7mm}\mathrm {s.t.}\hspace{6.3mm}d\bar{S}_r = 
e^{-\alpha \bar{\varphi }_r}\hat{b}(\bar{S}_re^{\alpha \bar{\varphi }_r})dr + 
e^{-\alpha \bar{\varphi }_r}\hat{\sigma }(\bar{S}_re^{\alpha \bar{\varphi }_r})dB_r, \\
&&\hspace{17mm}d\bar{W}_r = \frac{e^{\alpha \bar{\varphi }_r} - 1}{\alpha }d\bar{S}_r, \\
&&\hspace{20mm}\bar{S}_0 = \bar{s}, \ \ \bar{W}_0 = \bar{w}, 
\end{eqnarray*}
and
\begin{eqnarray*}
\overline {\mathcal {A}}_t(\varphi ) = \left\{ \left( \varphi - \int ^r_0\zeta _vdv\right) _{0\leq r\leq t}\ ; \ 
(\zeta _r)_{0\leq r\leq t} \in \mathcal {A}^\mathrm {SO}_t(\varphi )\right\} . 
\end{eqnarray*}
$\mathrm {(ii)}$ \ If $U$ is concave and $\hat{b} \leq 0$, then 
\begin{eqnarray}\label{eq_LL}
V^\mathrm {SO}_t(w, \varphi , s ; U) = U\left( w + \frac{1 - e^{-\alpha \varphi }}{\alpha }s\right) . 
\end{eqnarray}
\end{thm}

A proof is given in Section \ref {sec_proof_LL}. 
Note that assertion (ii) is the same as Theorem 3 in \cite {Lions-Lasry}, and 
in this case we can get the explicit form of the value function. 
The right side of (\ref {eq_LL}) equals $Ju(w,\varphi ,s)$ for $u(w, \varphi , s) = U(w)$ and 
the nearly optimal strategy for $V^{\mathrm {SO}}_t(w,\varphi ,s ; U) = V_t(w,\varphi ,s ; u)$ is given by (\ref {almost_block}).

\section{Examples}\label{sec_eg}
\setcounter{thm}{0}

In this section, we consider two examples of our model. 
Let $b(x) \equiv -\mu $ and $\sigma (x) \equiv \sigma $ for some 
constants $\mu , \sigma \geq 0$ and suppose $\tilde{\mu } = \mu - \sigma ^2/2 > 0$. 
We assume that the trader has a risk-neutral utility function 
$u(w, \varphi ,s) = u_{\mathrm {RN}}(w, \varphi ,s) = w$. 
Note that we can replace the stochastic control problem $V_t(w, \varphi ,s ; u_{\mathrm {RN}})$ with 
the deterministic control problem $f(t, \varphi )$, where 
\begin{eqnarray*}
f(t, \varphi ) &=& \sup _{(\zeta _r)_r\in \mathcal {A}^\mathrm {det}_t(\varphi )}\tilde{f}(t, \varphi ; (\zeta _r)_r), \\
\tilde{f}(t, \varphi ; (\zeta _r)_r) &=& 
\int ^t_0\zeta _r\exp \left( -\tilde {\mu }r - \int ^r_0g(\zeta _v) dv \right) dr, \\
\mathcal {A}^\mathrm {det}_t(\varphi ) &=& 
\{ (\zeta _r)_r\in \mathcal {A}_t(\varphi )\ ; \ 
(\zeta _r)_r \mathrm {\ is\ deterministic} \}. 
\end{eqnarray*}
Indeed, 
\begin{prop} \ \label{prop_deterministic} 
$V_t(w, \varphi ,s ; u_{\mathrm {RN}}) = w + sf(t, \varphi )$. 
\end {prop}

This is proved in Section \ref {proof_deterministic}. 
By Proposition \ref {prop_deterministic}, we see that 
\begin{eqnarray*}
\frac{\partial }{\partial s}V_t(w, \varphi , s ; u_{\mathrm {RN}}) = f(t, \varphi ) > 0, \ \ t, \varphi > 0, 
\end{eqnarray*}
which implies (\ref {cond_diff}). 

\subsection{Log-Linear Impact}\label{sec_linear_eg}

Set $g(\zeta ) = \alpha \zeta $ for $\alpha > 0$. 
The following theorem is a direct consequence of Theorem \ref {th_LL}(ii). 

\begin{thm}\label{th_eg} \ It holds that 
\begin{eqnarray}\label{eg}
V_t(w,\varphi ,s ; u_{\mathrm {RN}}) = w+\frac{1-e^{-\alpha \varphi }}{\alpha }s 
\end{eqnarray}
for all $t\in (0,1]$ and $(w,\varphi ,s)\in D$. 
\end{thm}

The right side of (\ref {eg}) 
converges to $w+\varphi s$ as $\alpha \downarrow 0$, which is the profit 
gained by choosing the so-called block liquidation execution strategy, that is,
by a trader selling all shares $\varphi $ at $t = 0$ when there is no MI. 
Theorem \ref {th_eg} implies that the optimal strategy in this case is 
to liquidate all shares, dividing infinitely within an infinitely short time at $t = 0$. 
This is almost the same as a block liquidation at the initial time, and 
the trader does not delay the execution time (although MI lowers the profit from the liquidation). 
Therefore, we cannot see any essential influence of the MI in this example.

\begin{remark}
We can also obtain the (nearly) optimal strategies in the cases of $\tilde {\mu } < 0$ and $\tilde{\mu } = 0$. 
When $\tilde{\mu } < 0$, the nearly optimal strategy is 
the (almost) block liquidation at the terminal time. 
When $\tilde{{\mu }} = 0$, each strategy in $\mathcal {A}^\mathrm {SO}_t(\varphi )$ makes 
the same profit: in other words, all the strategies in $\mathcal {A}^\mathrm {SO}_t(\varphi )$ are optimal. 
\end{remark}

\subsection{Log-Quadratic Impact}\label{sec_quad_eg}

In this subsection we consider the case of a strictly convex MI function. 
Set $g(\zeta ) = \alpha \zeta ^2$ for $\alpha > 0$. 
Note that $h(\zeta )/\zeta = 2\alpha > 0$, 
and thus the value function in this example is the unique viscosity solution of (\ref {HJB2}) 
with boundary conditions (\ref {bdd_cond}), by Theorems \ref {HJB_th} and \ref {unique_th}. 

As we will see, we can derive the explicit form of an optimal strategy 
when $\varphi $ is sufficiently small or large. 
However, when $\varphi $ is not sufficiently small, such a strategy has unbounded execution speed and 
is not subject to $\mathcal {A}_t(\varphi )$. 
Thus we extend the set of admissible strategies: 
\begin{eqnarray*}
\tilde {\mathcal {A}}_t(\varphi ) &=& \big\{ (\zeta _r)_{0\leq r\leq t}\ ; \ 
(\mathcal {F}_r)_r\mbox{-adapted}, \ \zeta _r\geq 0, \ \int ^t_0\zeta _rdr \leq \varphi  \\&&
\hspace{23mm}\mathrm {and} \ 
\sup _{(r,\omega )\in [0,t-\varepsilon ]\times \Omega }\zeta _r(\omega ) < \infty  \ \mathrm {for\ all} \ \varepsilon \in (0,t)\big\} , \\
\tilde {\mathcal {A}}^{\mathrm {det}}_t(\varphi ) &=& 
\{ (\zeta _r)_r\in \tilde{\mathcal {A}}_t(\varphi )\ ; \ 
(\zeta _r)_r \mathrm {\ is\ deterministic} \} 
\end{eqnarray*}
to allow unbounded execution speed at $t$. 
We can see that the value of $f(t, \varphi )$ \vspace{1mm}
does not change by replacing $\mathcal {A}^{\mathrm {det}}_t(\varphi )$ with $\tilde {\mathcal {A}}^{\mathrm {det}}_t(\varphi )$. 
Indeed, for each $(\zeta _r)_r\in \tilde {\mathcal {A}}^\mathrm {det}_t(\varphi )$, 
the integrability of $\zeta _r$ on $[0, t]$ 
$\left(\mathrm {i.e.} \int ^t_0\zeta _rdr\leq \varphi < \infty \right) $, 
the dominated convergence theorem, and the continuity of $f(t, \varphi )$ in $t$ 
(this is obtained by Theorem \ref {conti}(i) and Proposition \ref {prop_deterministic}) imply 
\begin{eqnarray}\nonumber 
\tilde{f}(t, \varphi ; (\zeta _r)_r)  &=& 
\lim _{\varepsilon \rightarrow 0}\tilde{f}(t - \varepsilon , \varphi ; (\zeta _r)_r)\\\label{tilde_eq}
&\leq & 
\lim _{\varepsilon \rightarrow 0}f(t-\varepsilon , \varphi ) = f(t, \varphi ). 
\end{eqnarray}
So we get 
\begin{eqnarray*}
f(t, \varphi ) = \sup _{(\zeta _r)_r\in \tilde{\mathcal {A}}^{\mathrm {det}}_t(\varphi )}
\tilde{f}(t, \varphi ; (\zeta _r)_r). 
\end{eqnarray*}
Thus, we can also restrict the set of admissible strategies of 
$V_t(w, \varphi ,s  ; \allowbreak  u_{\mathrm RN})$ 
to $\tilde {\mathcal {A}}^{\mathrm {det}}_t(\varphi )$ by 
Proposition \ref {prop_deterministic}. 

We define the functions $\hat{v}^i(t, w, \varphi , s)$ and $\hat{\zeta }^i_r$, $i = 1, 2$, by 
\begin{eqnarray*}
\hat{v}^1(t,w,\varphi ,s) = w+
\frac{s\sqrt{1-e^{-2\tilde{\mu } t}}}{2\sqrt{\alpha \tilde{\mu } }}, \ \ 
\hat{\zeta }^1_r = 
\sqrt{\frac{\tilde{\mu }}{\alpha (1-e^{-2\tilde{\mu } (t-r)})}} 
\end{eqnarray*}
and 
\begin{eqnarray*}
\hat{v}^2(t,w,\varphi ,s) = w + \frac{s}{2\sqrt{\alpha \tilde{\mu }}}
( 1 - e^{-2\sqrt{\alpha \tilde{\mu }}\varphi }), \ \ 
\hat{\zeta }^2_r = 
\sqrt {\frac{\tilde{\mu }}{\alpha }}1_{[0,\varphi \sqrt{\alpha/\tilde{\mu }}]}(r). 
\end{eqnarray*}
Moreover we set 
\begin{eqnarray}
\hat{\Phi }^1(t)= \frac{\mathrm {arctanh}\sqrt{1-e^{-2\tilde{\mu }t}}}{\sqrt{\alpha \tilde{\mu }}}, \ \ 
\hat{\Phi }^2(t)= \sqrt{\frac{\tilde{\mu }}{\alpha }}t. 
\end{eqnarray}
Then we have the following theorem: 

\begin{thm} \ \label{th_eg2}\\
$(\mathrm {i})$ \ If $ \varphi \geq \hat{\Phi }^1(t)$, then 
$V_t(w, \varphi ,s ; u_{\mathrm {RN}}) =\hat{v}^1(t, w, \varphi , s)$ and 
$(\hat{\zeta }^1_r)_r$ is an optimal strategy. \\
$(\mathrm {ii})$ \ If $\varphi \leq \hat{\Phi }^2(t)$, then 
$V_t(w, \varphi ,s ; u_{\mathrm {RN}}) =\hat{v}^2(t, w, \varphi , s)$ and 
$(\hat{\zeta }^2_r)_r$ is an optimal strategy. 
\end{thm}

\begin{proof}
Let $(\hat{W}^i_r, \hat{\varphi }^i_r, \hat{S}^i_r)_r = \Xi _t(w, \varphi , s ; (\hat{\zeta }^i_r)_r)$ for $i = 1, 2$. 
Straightforward calculation shows that 
$\E [\hat{W}^i_t] = \hat{v}^i(t, w, \varphi , s)$. 
Then we have $\hat{v}^i(t, w, \varphi ,s) \leq V_t(w, \varphi ,s ; u_{\mathrm {RN}})$. 
Since $\hat{v}^i$ satisfies (\ref {HJB}) at $(t, w, \varphi , s)$, we see 
that $\hat{v}^i(t, w, \varphi ,s) \geq V_t(w, \varphi ,s ; \allowbreak u_{\mathrm {RN}})$ 
by Theorem 5.2.1 in \cite{Nagai}, 
thus fulfilling the assertion. 
\end{proof}

This theorem implies that the form of optimal strategies and value functions varies, depending on 
the amount of the security holdings $\varphi $. 
If a trader has a small amount of securities, then we have case (ii) and the optimal strategy is to sell 
the entire holdings until the time $\varphi \sqrt{\alpha / \tilde{\mu }}$. 
If a trader has a large amount, then we have case (i) and the trader cannot finish the selling. 

We do not have an explicit form for $V_t(w, \varphi ,s ; u_{\mathrm {RN}})$ on the whole space, 
so we try to solve this example numerically. 
By Proposition \ref {prop_deterministic}, it suffices to consider the numerical calculation of $f(t, \varphi )$. 
Moreover, $f(t, \varphi )$ is approximated by $f^n_{[nt]}(\varphi )$ with large $n$, where 
\begin{eqnarray}\label{f_n_k1}
f^n_k(\varphi ) &=& \sup _{(\psi ^n_l)_l\in \mathcal {A}^{n, \mathrm {det}}(\varphi )}
\sum ^{k-1}_{l = 0}\psi ^n_l\exp \left( -\tilde{\mu }\times \frac{l}{n} - n\alpha \sum ^{l}_{m=0}(\psi ^n_m)^2\right) , \\\label{f_n_k3}
\mathcal {A}^{n, \mathrm {det}}_k(\varphi ) &=& 
\left\{ (\psi ^n_l)^{k-1}_{l=0}\subset [0, \varphi ]^k\ ; \ 
\sum ^{k-1}_{l = 0}\psi ^n_l \leq \varphi \right\} . 
\end{eqnarray}
In fact, 
the convergence of $f^n_{[nt]}$ to $f_t$ is given by the same proof as Theorem \ref {converge} in Section \ref {sec_discrete}. 
$f^n_k(\varphi )$ corresponds to a nonlinear optimisation problem with $k$ variables. 
We solve it numerically by the sequential quadratic programming method. 

It can be numerically verified that the convergence of $f^n_{[nt]}$ takes place before $n = 500$. 
Thus, we set $n = 500$ below and we regard $f^{500}_{[500t]}(\varphi )$ as a precise approximation of $f(t, \varphi )$. 
We set other parameters as $w = 0, s = 1, \alpha = 0.01$, and $\tilde {\mu }= 0.05$.

\begin{figure*}[tb]
\begin{center}
\includegraphics[height = 2.432cm,width=3.8cm]{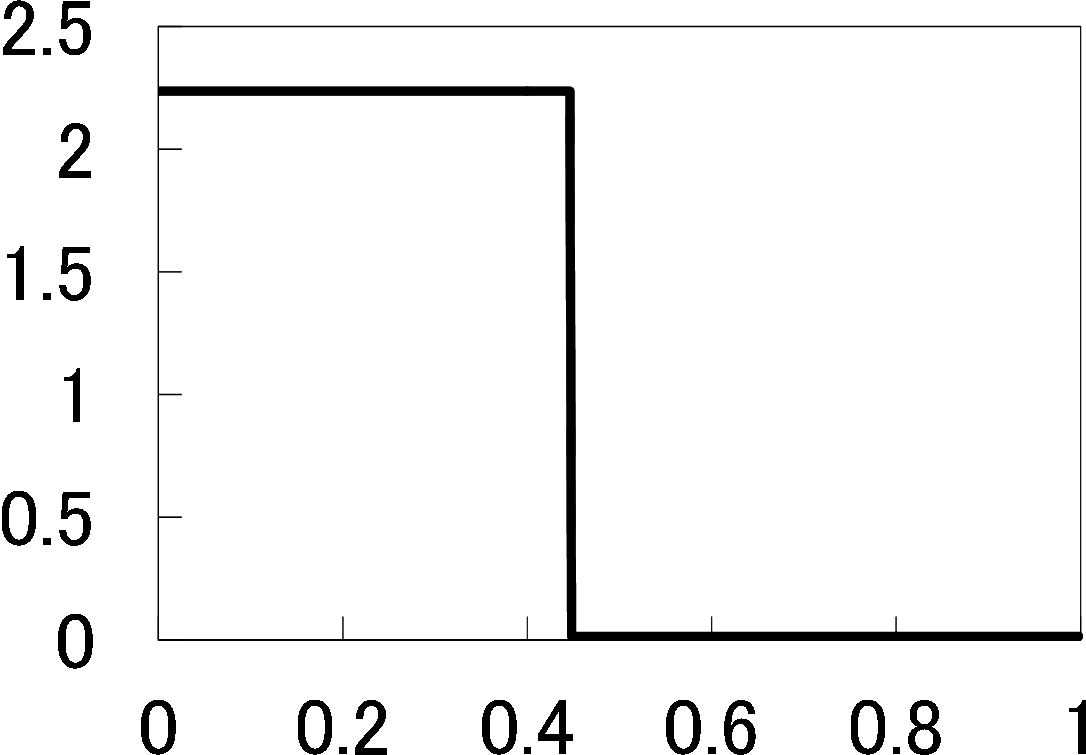}
\includegraphics[height = 2.432cm,width=3.8cm]{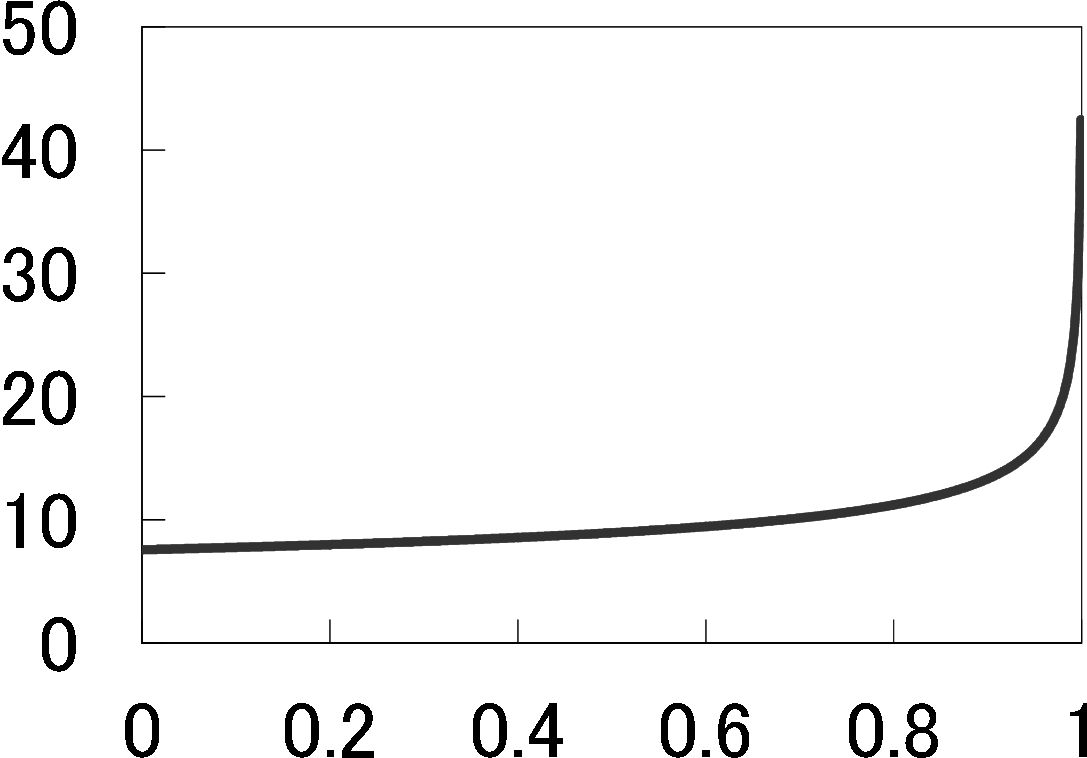}
\includegraphics[height = 2.432cm,width=3.8cm]{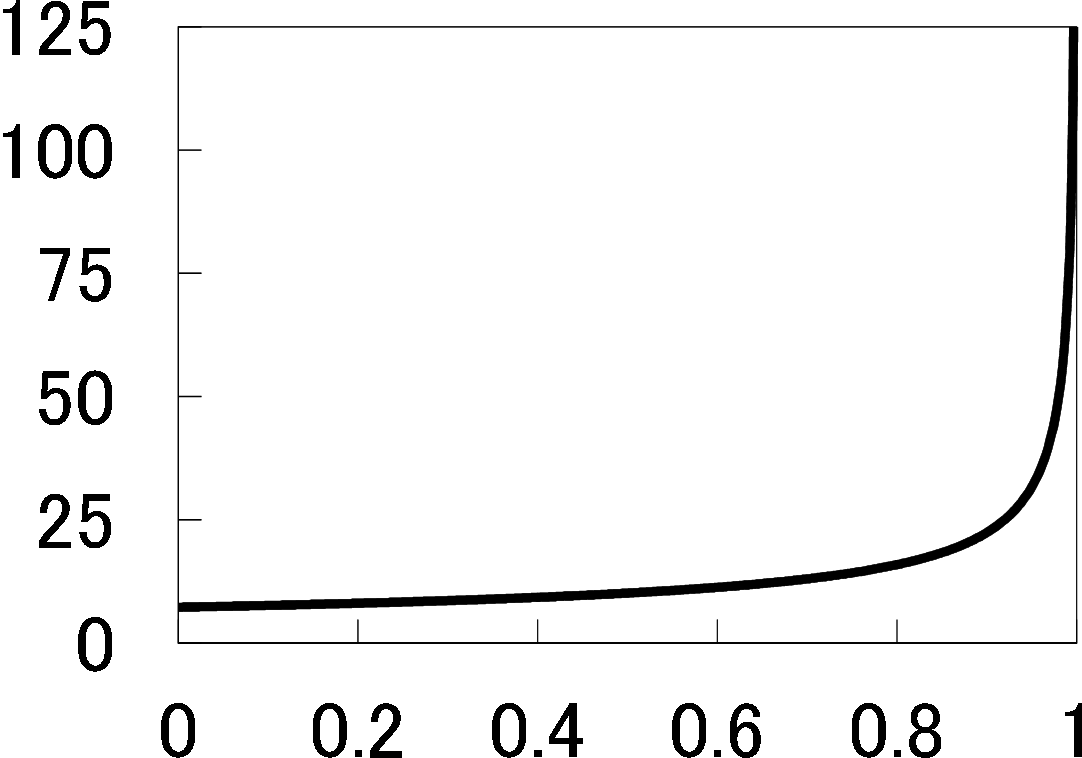}
\caption{The forms of optimal execution strategies $(\zeta _r)_r$. 
Horizontal axis is time $r$. 
Left: $\varphi = 1$. Centre: $\varphi = 10$. Right: $\varphi = 100$. 
In the centre graph, $(\zeta _r)_r$ was calculated numerically. }
\label{fig_1}
\end{center}
\end{figure*}

\begin{figure*}[tb]
\begin{center}
\includegraphics[height = 2.432cm,width=3.8cm]{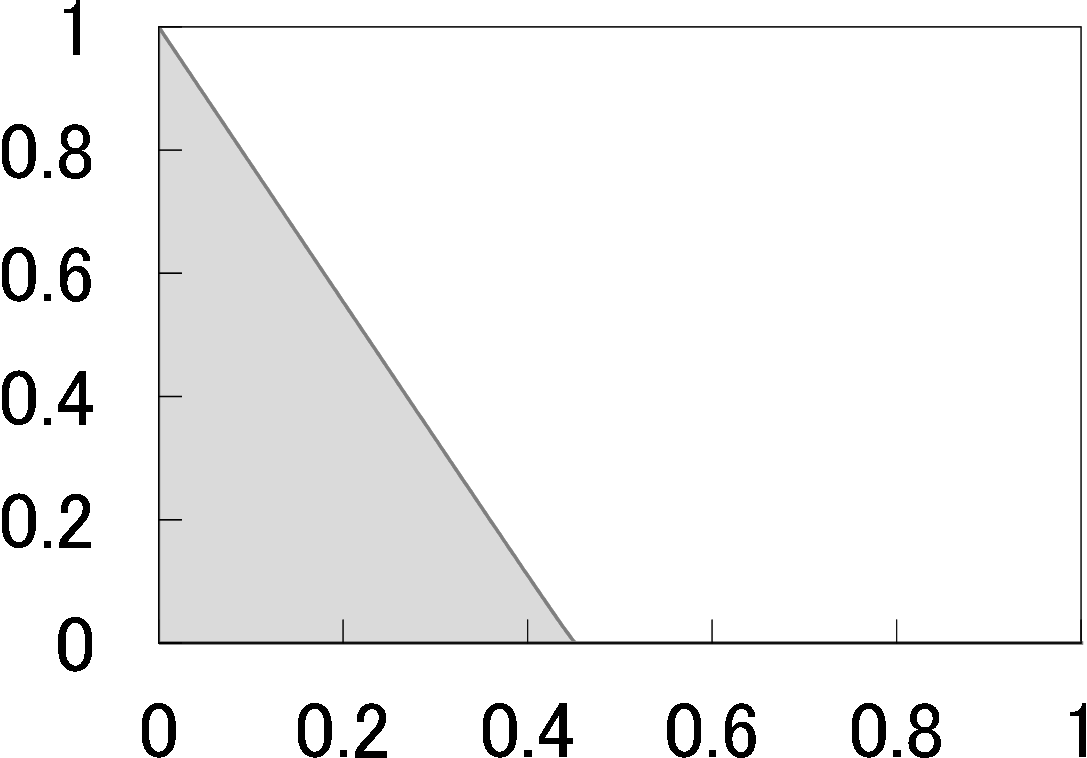}
\includegraphics[height = 2.432cm,width=3.8cm]{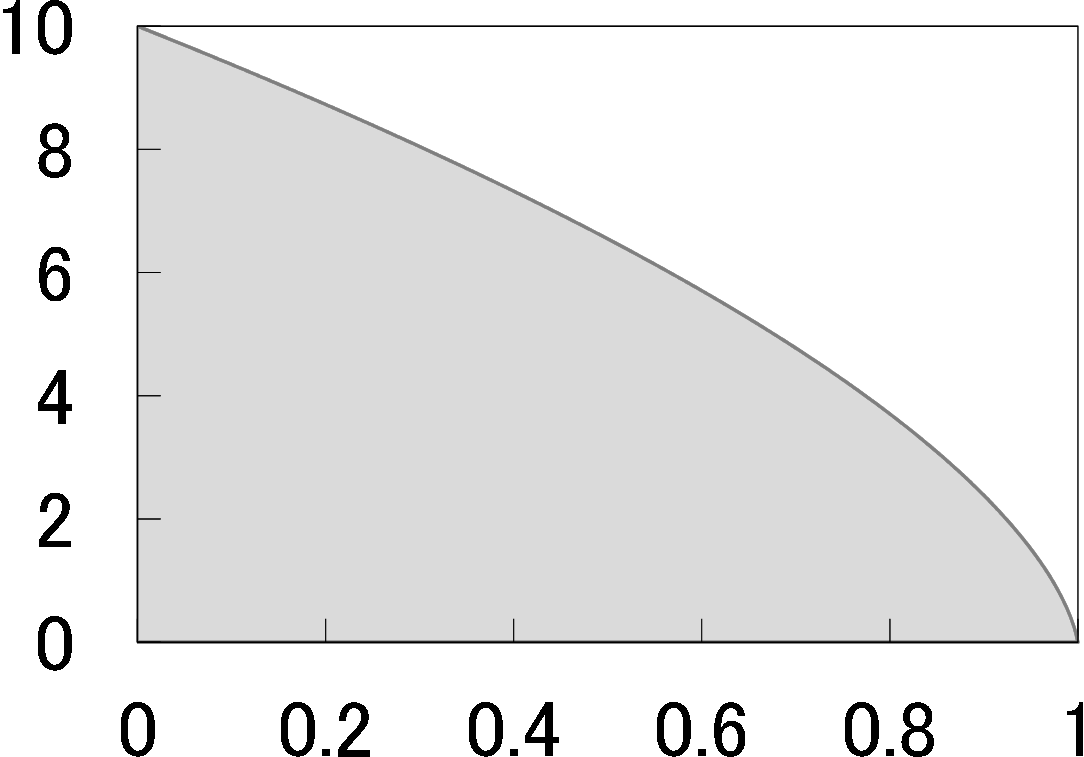}
\includegraphics[height = 2.432cm,width=3.8cm]{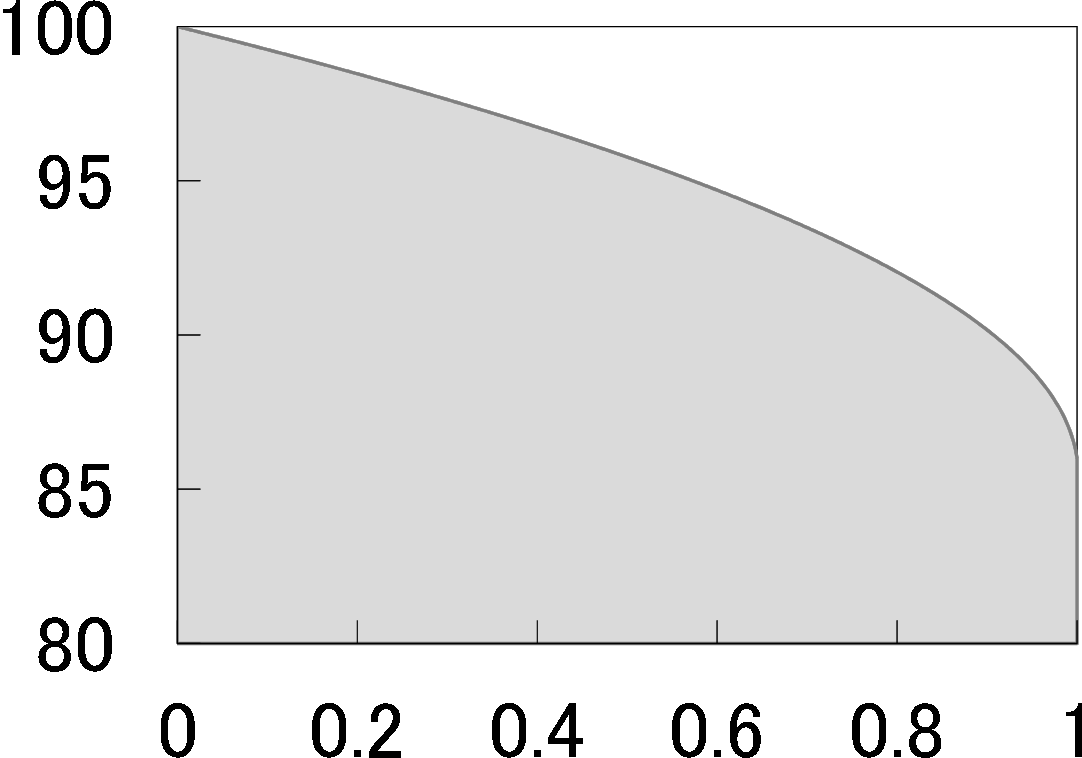}
\caption{The forms of the amount of security holdings $(\varphi _r)_r$ corresponding to optimal strategies. 
Horizontal axis is time $r$. 
Left: $\varphi = 1$. Centre: $\varphi = 10$. Right: $\varphi = 100$. 
In the centre graph, $(\varphi _r)_r$ was calculated numerically. }
\label{fig_2}
\end{center}
\end{figure*}

\begin{figure*}[htb]
 \begin{minipage}{0.5\hsize}
\begin{center}
\includegraphics[width=6cm]{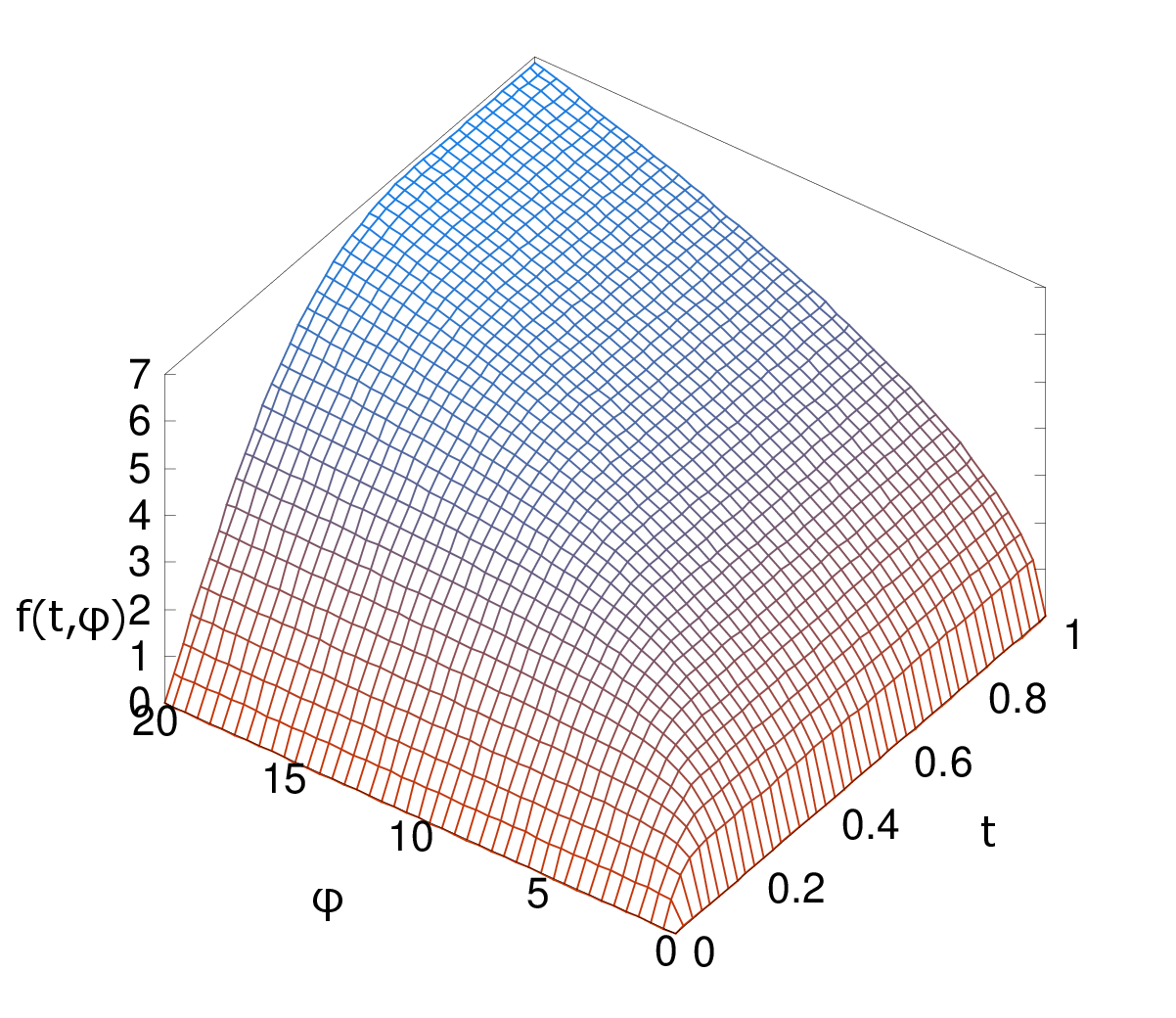}
\caption{The form of the function $f(t, \varphi )$.}
\label{fig_3}
\end{center}
 \end{minipage}
 \begin{minipage}{0.5\hsize}
\begin{center}
\includegraphics[width=6cm]{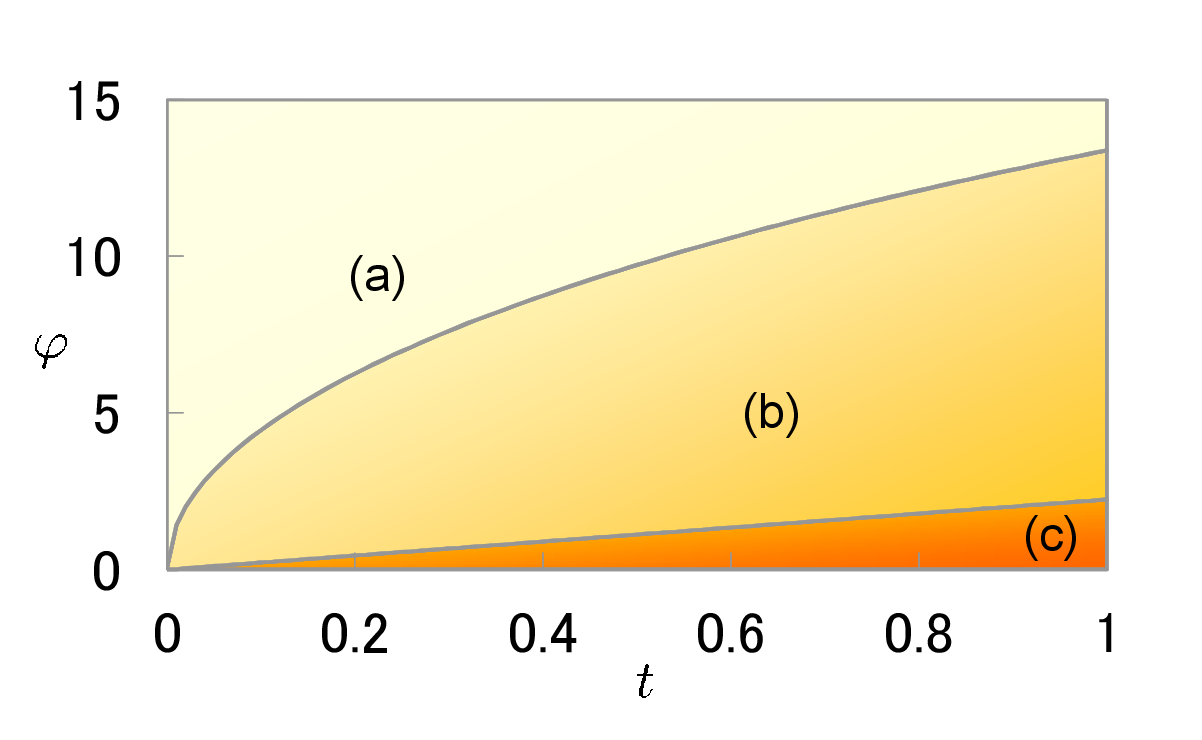}
\caption{The region of pairs $(t, \varphi )$. The region (a) (resp. (c)) corresponds to Theorem \ref {th_eg2} (i) (resp. (ii)). }
\label{fig_4}
\end{center}
 \end{minipage}
\end{figure*}

Figure \ref {fig_1} describes the form of the execution strategies and Figure \ref {fig_2} describes 
the form of the corresponding processes of the amount of a security with $\varphi = 1, 10$, and $100$. 
We also get the form of the function $f(t, \varphi )$ of Proposition \ref {prop_deterministic} numerically, 
as described in Figure \ref {fig_3}. 
If a pair $(t, \varphi )$ is in the range (a) of Figure \ref {fig_4}, then we have 
$f(t, \varphi ) = \sqrt{1-e^{-2\tilde
{\mu } t}} / (2\sqrt{\alpha \tilde{\mu } })$, 
and if $(t, \varphi )$ is in the range (c), we have 
$f(t, \varphi ) = 
(1 - e^{-2\sqrt{\alpha \tilde{\mu }}\varphi }) / (2\sqrt{\alpha \tilde{\mu }})$. 
We have not had the form of $f(t, \varphi )$ analytically when $(t, \varphi )$ is in the range (b). 

Note that in case (i) we can also construct a nearly optimal strategy with the sell-off condition. 
Let $\hat{\zeta }^{1, \delta }_r = \hat{\zeta }^1_r\ (r \leq t-\delta ), \ 
(\varphi - \hat{\varphi }_{t-\delta })/\delta \ (t-\delta < r \leq t)$, where 
\begin{eqnarray*}
\hat{\varphi }_{t-\delta } = 
\frac{\mathrm{arctanh}\sqrt{1-e^{-2\tilde{\mu t}}} - 
\mathrm{arctanh}\sqrt{1-e^{-2\tilde{\mu }\delta }}}
{\sqrt{\alpha \tilde {\mu }}}. 
\end{eqnarray*}
Then $(\hat{\zeta }^{1, \delta }_r)_r\in \mathcal {A}^\mathrm {SO}_t(\varphi )$ and 
the corresponding expected profit $\E [\hat{W}^\delta _t]$ converges to 
$V_t(w, \varphi , s ; u_\mathrm {RN})$ as $\delta \rightarrow 0$.

\section{Concluding Remarks}\label{sec_summary}
\setcounter{thm}{0}

In this paper we studied the optimal execution problem when MI is considered. 
We mainly considered the case where the MI function is convex. 
This was done for both mathematical and financial reasons.
In a Black--Scholes type market, 
the optimal execution strategy of a risk-neutral trader is block liquidation when there is no MI. 
As we saw in Section \ref {sec_eg}, 
the form of the optimal strategy changes when MI is log-quadratic. 
In contrast, when MI is not convex, and especially when it is log-linear, 
the trader's optimal strategy is almost block liquidation. 

In the real market, however, many traders take their time selling, despite 
recognition that the MI is concave. 
One reason may be that the trader has a risk-averse utility function. 
We surmise another reason: 
the existence of a temporary (or transient) impact (see Remark \ref {rem_temporary}). 
Our examples treat only permanent impact, 
but we can also consider the case where MI disappears as time passes by price recovery effects: 
if the process of security prices follows some mean-reverting process, such as an Ornstein--Uhlenbeck process, 
then we may deal with the optimisation problem with MI and price recovery. 
We study such a case in \cite {Kato4}. 

It is also meaningful to characterise the value function as the solution of the corresponding HJB. 
We have shown that the value function is a viscosity solution under some strong assumptions. 
Such assumptions would  not be necessary if we considered only bounded strategies, 
but the control region of our model is unbounded. 
We avoid this difficulty by supposing (\ref {cond_diff}) 
which is satisfied in financially natural settings. 

In trading operations, 
the trader should execute trades while considering fluctuations of the price of other assets
(e.g., rebalancing an index fund). 
In \cite{Ishitani}, a multidimensional version of this model was studied to consider such a case. 
However, in the case of rebalancing, 
it is necessary to consider not only selling securities but also buying them. 
We should carefully formulate models of optimal execution 
so that no trader gets a free lunch when MI is large. 

The complete solution of our example in Section \ref {sec_quad_eg} is another remaining task. 
This is a representative example where a trading policy is strongly influenced by MI, and 
it would be interesting to solve this completely in future research. 

\section{Appendix A: Derivation of the Continuous-Time Model from the Discrete-Time Models}\label{sec_discrete}
\setcounter{thm}{0}

Here we construct a discrete-time model of an optimal execution with time interval $1/n$ ($n = 1, 2, 3, \ldots $). 
As in Section \ref {section_Model}, 
we prepare a filtered space $(\Omega ,\mathcal {F}, (\mathcal {F}_t)_{0\leq t\leq 1}, \allowbreak P)$ 
satisfying the usual condition and a one-dimensional $(\mathcal {F}_t)_t$-Brownian motion 
$(B_t)_t$. 
We assume that there are cash and a security 
and that the risk-free rate is equal to zero (i.e., the price of cash is $1$). 
We consider a single trader who has $\varphi $ shares of the security at the initial time and tries to sell them. 

Now we consider the situation of trading at each execution time 
$t = 0, 1/n, \allowbreak 2/n, \ldots ,$ 
and describe the effect of the trader's liquidation. 
For $l=0, \ldots , n$, we denote by $S^n_l$ the price of the security at time $l/n$ and 
$X^n_l = \log S^n_l$. 
Let $s > 0$ be the initial price (i.e., $S^n_0 = s$) and $X^n_0 = \log s$. 
If the trader sells $\psi ^n_l$ at time $l/n$, the log-price changes to 
$X^n_l-g_n(\psi ^n_l)$, where 
$g_n : [0,\infty )\longrightarrow [0,\infty )$ is a non-decreasing and continuously 
differentiable function which satisfies $g_n(0) = 0$. 
The function $g_n$ denotes the MI function in the discrete-time model: 
$g_n(\psi ^n_l)$ implies the impact of the liquidation of $\psi ^n_l$ shares 
for the log-price of the security. 
So the security price $S^n_l$ decreases to $S^n_l\exp (-g_n(\psi ^n_l))$ by the liquidation. 
Then the trader gets $\psi ^n_lS^n_l\exp (-g_n(\psi ^n_l))$ in cash as the proceeds of the liquidation. 
Thus, if we denote by $W^n_l$ (resp. $\varphi ^n_l$) the amount of the cash holdings (resp. security holdings), 
then we have the following relations 
\begin{eqnarray}\label{W_varphi}
W^n_{l+1} = W^n_l+\psi ^n_lS^n_l\exp (-g_n(\psi ^n_l)), \ \ 
\varphi ^n_{l+1} = \varphi ^n_l - \psi ^n_l. 
\end{eqnarray}
The former means the increase of the cash holdings and the latter means 
the decrease of the security holdings. 

After trading at time $l/n$, $X^n_{l+1}$ and $S^n_{l+1}$ are given by 
\begin{eqnarray}\label{fluctuate_X}
X^n_{l+1}=Y\Big (\frac{l+1}{n} ; \frac{l}{n}, X^n_l-g_n(\psi ^n_l)\Big ), 
\ \ S^n_{l+1} = \exp (X^n_{l+1}), 
\end{eqnarray}
where $Y(t ; r,y)$ is the solution of the SDE 
\begin{eqnarray}\label{SDE_Y}
\left\{
\begin{array}{ll}
 	dY(t; r,y) = \sigma (Y(t; r,y))dB_t+b(Y(t; r,y))dt,& t\geq r,	\\
 	\hspace{2mm}Y(r; r,y) = y.&
\end{array}
\right.
\end{eqnarray}
Note that if the trader makes no liquidation, then the unaffected log-price $X^n_l$ 
coincides with $Y(l/n ; 0, x)$. 
The first equation of (\ref {fluctuate_X}) describes the fluctuation of the log-price 
as time passes from $l/n$ (with the affected log-price $X^n_l - g_n(\psi ^n_l)$) to $(l+1)/n$. 

Here we give a class of our execution strategies. 
Let $\mathcal {A}^n_k(\varphi )$ be the set of strategies $(\psi ^n_l)^{k-1}_{l=0}$ such that 
$\psi ^n_l$ is $\mathcal {F}_{l/n}$-measurable, $\psi ^n_l\geq 0$ for any $l=0,\ldots ,k-1$, 
and $\sum ^{k-1}_{l=0}\psi ^n_l\leq \varphi $. 
We call $\mathcal {A}^n_k(\varphi )$ the set of admissible strategies. 
An admissible strategy is the sequence of random variables (liquidation volumes) $(\psi ^n_l)_l$ 
which are constructed by only selling with short-sale constraint 
(the trader does not buy and does not sell short). 

At the end of the time interval $[0,1]$ (i.e. $1 = n/n$), 
the trader has the amount of cash $W^n_n$ and 
the amount of the security $\varphi ^n_n$, which are determined by 
(\ref {W_varphi}) for each $l=0,\ldots ,n-1$ and initial values $W^n_0 = w,\ \varphi ^n_0 = \varphi $. 

Now we define our value function in the discrete-time model. 
For $(w, \varphi , s)\in D$, $k = 0, \ldots , n$ and $u\in \mathcal {C}$ 
(the definitions of $D$ and $\mathcal {C}$ are the same as Section \ref {section_Model}), set 
\begin{eqnarray}\label{def_discrete_value}
V^n_k(w,\varphi ,s ; u) = \sup _{(\psi ^n_l)^{k-1}_{l=0}\in \mathcal {A}^n_k(\varphi )}
\E [u(W^n_k,\varphi ^n_k, S^n_k)],
\end{eqnarray}
subject to (\ref {W_varphi}) and (\ref {fluctuate_X}) for $l=0,\ldots ,k-1$ and 
$(W^n_0,\varphi ^n_0,S^n_0) = (w,\varphi ,s)$ when $s > 0$. 
In the case of $s = 0$, we define $V^n_k(w, \varphi, 0 ; u) = u(w, \varphi , 0)$. 

Now we assume condition [A]: \vspace{2mm}\\
\hspace{0mm}[A] \ 
$\lim _{n\rightarrow \infty }\sup _{\psi \in [0,\Phi _0]}
\Big |\frac{d}{d\psi }g_n(\psi )-h(n\psi )\Big | = 0$. \vspace{2mm}

Recall that 
$h(\zeta ) = g'(\zeta )$ is a non-negative, non-decreasing, and continuous 
function (see Section \ref {section_Model}). 
Under condition [A], we see that $\varepsilon _n\longrightarrow 0$, where 
\begin{eqnarray}\label{conv_g}
\varepsilon _n = \sup _{\psi \in (0,\Phi _0]}
\Big| \frac{g_n(\psi )}{\psi }-\frac{g(n\psi )}{n\psi }\Big| . 
\end{eqnarray}
This implies the relation between the MI function in the discrete-time model 
and the one in the continuous-time model. 
The condition [A] roughly means the $C^1$-convergence of $g_n$ to $g$. 
Under [A], we can prove the following theorem. 

\begin{thm} \ \label{converge}
For any $(w,\varphi ,s)\in D$, $t\in [0,1]$ and $u\in \mathcal {C}$,
\begin{eqnarray}\label{th_2}
\lim _{n\rightarrow \infty }V^n_{[nt]}(w,\varphi ,s;u) = V_t(w,\varphi ,s;u), 
\end{eqnarray}
where $[nt]$ is the greatest integer less than or equal to $nt$. 
\end{thm}

A proof is given in Section \ref {proof_of_converge}. 
Theorem \ref {converge} implies the convergence of 
the discrete-time value function to the continuous-time one. 
In other words, our model in Section \ref {section_Model} 
is characterised as the limit of the discrete-time models.

\section{Appendix B: Proofs}\label{Proofs}
\setcounter{thm}{0}

\subsection{Preliminaries}\label{pre}

We introduce some lemmas used to prove our main results. 

\begin{lem}\label{cond_of_X2} For any $m\in \Bbb {N}$ there is a constant 
$C>0$ depending only on $b, \sigma $ and $m$ such that 
$\E [\hat{Z}(s)^m]\leq Cs^m$, where $\hat{Z}(s) = \sup _{0\leq t\leq 1}Z_t(s)$ and 
$Z_t(s)$ is defined in Theorem \ref{unique_th}. 
\end{lem}

\begin{proof} We may assume $s>0$. By the definition of $\hat{Z}(s)$, we have 
\begin{eqnarray*}
\E [\hat{Z}(s)^m]\leq s^m\E [\sup _{t\in [0,1]}\tilde{Z}_t] , 
\end{eqnarray*}
where $(\tilde{Z}_t)_t$ is given by $\tilde{Z}_0 = 1$ and 
\begin{eqnarray*}
d\tilde{Z}_t = m\tilde {Z}_t\sigma \left (\frac{1}{m}\log \tilde {Z}_t\right )dB_t + 
m\tilde {Z}_t\left\{ b\left (\frac{1}{m}\log \tilde {Z}_t\right ) + \frac{m}{2}\sigma \left (\frac{1}{m}\log \tilde {Z}_t\right )^2\right\} dt. 
\end{eqnarray*}
Using Corollary 2.5.10 in \cite{Krylov} for the process
$(\tilde {Z}_t)_t$, we have the assertion. 
\end{proof}

\begin{lem} \ \label{conti_u} 
Let $\Gamma _k$, $k\in \Bbb {N}$, be sets, $u\in \mathcal {C}$ and 
let $(W^i_{k,\gamma }, \varphi ^i_{k,\gamma }, S^i_{k,\gamma })\in D$, 
$\gamma \in \Gamma _k$, $k\in \Bbb {N}$, $i=1,2$, be random variables. 
Let $m_u\in \Bbb {N}$ be as in $(\ref {growth_C})$. Suppose 
\begin{eqnarray*}
\lim _{k\rightarrow \infty }\sup _{\gamma \in \Gamma _k} 
\E [|W^1_{k,\gamma }-W^2_{k,\gamma }| + |\varphi ^1_{k,\gamma }-\varphi ^2_{k,\gamma }| + 
|S^1_{k,\gamma }-S^2_{k,\gamma }|] = 0 
\end{eqnarray*}
and 
$\sum ^2_{i=1}\sup _{k\in \Bbb {N}}\sup _{\gamma \in \Gamma _k}
\E [(W^i_{k,\gamma })^{4m_u}+(S^i_{k,\gamma })^{4m_u}] < \infty $. 
Then
\begin{eqnarray*}
\lim _{k\rightarrow \infty }\sup _{\gamma \in \Gamma _k}
\big| \E [u(W^1_{k,\gamma }, \varphi ^1_{k,\gamma }, S^1_{k,\gamma })] - 
\E [u(W^2_{k,\gamma }, \varphi ^2_{k,\gamma }, S^2_{k,\gamma })]\big| = 0. 
\end{eqnarray*}
\end{lem}

This lemma is obtained by standard arguments using the Chebyshev inequality and 
the uniform continuity of $u(w,\varphi ,s)$ on $D_R$ for any $R>0$, 
where $D_R = [-R,R]\times [0,\Phi _0]\times [0,R]$. 

Using the Burkholder--Davis--Gundy inequality and the H\"older inequality, 
we have the following lemma: 

\begin{lem}\label{cond_of_X1} \ Let $t\in [0,1]$, $\varphi \geq 0$, $x\in \Bbb {R}$, 
$(\zeta _r)_{0\leq r\leq t}\in \mathcal {A}_t(\varphi )$ and let 
$(X_r)_{0\leq r\leq t}$ be given by $(\ref {SDE_X})$ with $X_0 = x$. 
Then there is a constant $C>0$ depending only on $b$ and $\sigma $, such that 
\begin{eqnarray*}
\E \Big [\sup _{r\in [r_0,r_1]}
\Big |X_r-X_{r_0}+\int ^r_{r_0}g(\zeta _v)dv\Big |^4\Big ] \leq C(r_1-r_0)^2, \ \ 
0\leq r_0\leq r_1\leq t. 
\end{eqnarray*}
\end{lem}

\begin{lem}\label{cond_of_X3} \ 
Let $t\in [0,1]$, $\varphi \geq 0$, $x\in \Bbb {R}$, 
$(\zeta _r)_{0\leq r\leq t}, (\zeta '_r)_{0\leq r\leq t}
\in \mathcal {A}_t(\varphi )$ and suppose 
$(X_r)_{0\leq r\leq t}$ $($resp., $(X'_r)_{0\leq r\leq t}$$)$ 
is given by $(\ref {SDE_X})$ with $(\zeta _r)_r$ 
$($resp., $(\zeta '_r)_r$$)$ and $X_0 = x \leq X'_0$. 
Suppose $\zeta _r\leq \zeta '_r$ for any $r\in [0,t]$ almost surely. 
Then $X_r\geq X'_r$ for any $r\in [0,t]$ almost surely. 
In particular, we have $\exp (X_r)\leq \hat{Z}(e^x)$. 
\end{lem}

This lemma is obtained by the same arguments as in the proof of 
Proposition 5.2.18 in \cite {Karatzas-Shreve}. 
Lemmas \ref {cond_of_X2} and \ref {cond_of_X3} imply 

\begin{lem} \ \label{nondecreasing_polygrowth}
For $t\in [0,1]$ and $u\in \mathcal {C}$, 
$V_t( w, \varphi ,s ; u)$ is non-decreasing in $w, \varphi $ and $s$, 
\noindent and has polynomial growth rate with respect to $w$ and $s$. 
\end{lem}

\subsection{Strategy-Restricted Value Functions}\label{S-R}

We prepare strategy-restricted value functions 
to prove Theorems \ref {conti} and \ref {semi}. For $L>0$, we define 
\begin{eqnarray*}
\mathcal {A}^L_t(\varphi ) &=& \{ (\zeta _r)_{0\leq r\leq t}\in \mathcal {A}_t(\varphi )\ ; \ 
\sup _{r,\omega }|\zeta _r(\omega )|\leq L \} , \\
V^L_t(w,\varphi ,s ; u) &=& \sup _{(\zeta _r)_{r\leq t}\in \mathcal {A}^L_t(\varphi )}
\E [u(W_t, \varphi _t, S_t)]. 
\end{eqnarray*}
We easily see that 
$V_t(w,\varphi ,s ; u) = 
\sup _{L>0}\allowbreak V^L_t(w,\varphi ,s ; u)$. 

Now we consider the continuity of $V^L_t(w,\varphi ,s ; u)$. 
Our purpose here is to prove the following proposition: 

\begin{prop} \ \label{conti_L}
$V^L_t(w,\varphi ,s ; u)$ is continuous with respect to $(t,w,\varphi ,s)\in [0,1]\times D$. 
\end{prop}

To prove Proposition \ref {conti_L}, we prove the following lemmas: 

\begin{lem}\label{conti_wphis_L} \ 
For any $(w,\varphi ,s)\in D$ and $t\in [0,1]$ 
\begin{eqnarray*}
\lim _{(w',\varphi ', s')\rightarrow (w,\varphi ,s)}\sup _{L>0}
|V^L_t(w',\varphi ',s' ; u)-V^L_t(w,\varphi ,s ; u)|=0. 
\end{eqnarray*} 
\end{lem}

\begin{proof} 
Let $R>0$ and $(w,\varphi ,s), (w',\varphi ',s')\in D_R$. We may assume $s'>0$. 
Take any $(\zeta _r)_{r\leq t}\in \mathcal {A}^L_t(\varphi )$. 
Let $\rho =\inf \{ r>0\ ; \ \int ^r_0\zeta _vdv>\varphi \wedge \varphi '\} \wedge t$ and 
$\zeta '_r = \zeta _r1_{\{r\leq \rho \}}$. 
Then $(\zeta '_r)_{r\leq t}\in \mathcal {A}^L_t(\varphi ')$. 
Let $(W_r, \varphi _r, S_r)_{r\leq t} = \Xi _t(w,\varphi ,s ; (\zeta _r)_r)$ and 
$(W'_r, \varphi '_r, S'_r)_{r\leq t} = \Xi _t(w',\varphi ',s' ; (\zeta '_r)_r)$. 
Moreover, let us define $(\tilde{S}'_r)_{r\leq t}$ by 
\begin{eqnarray*}
d\tilde{S}'_r = \hat{\sigma }(\tilde{S}'_r)dB_r+\hat{b}(\tilde{S}'_r)dr-
g(\zeta  _r)\tilde{S}'_rdr, \ \ 
\tilde{S}'_0 = s'. 
\end{eqnarray*}
Then Lemma \ref {cond_of_X3} implies $S'_r\geq \tilde{S}'_r$ for any $r\in [0,t]$ 
almost surely. Thus 
\begin{eqnarray}\label{t_1}
\E [u(W_t,\varphi _t,S_t)]-V_t(w',\varphi ',s' ; u) 
\leq 
\E [|u(W_t,\varphi _t,S_t)-u(W'_t,\varphi '_t,\tilde{S}'_t)|] . \hspace{5mm}
\end{eqnarray}
By a simple calculation we get 
\begin{eqnarray*}
|W_t-W'_t|\leq 
|w-w'|+\hat{Z}(s)|\varphi -\varphi '|+\Phi _0\sup _{r\in [0,t]}|S_r-\tilde{S}'_r|
\end{eqnarray*}
and $|\varphi _t-\varphi '_t|\leq |\varphi -\varphi '|$. 
Moreover, Theorem 3.2.7 in \cite{Nagai} and Lemma \ref {cond_of_X2} imply 
\begin{eqnarray*}
\E [\sup _{r\in [0,t]}|S_r-\tilde{S}'_r|] \leq \left\{
                                               \begin{array}{ll}
                                                C_0s'	& (s=0)	\\
                                                C_0|\log s-\log s'|	& (s>0)
                                               \end{array}
                                               \right.
\end{eqnarray*}
for some $C_0>0$ depending on only $b,\sigma $ and $R$. 
Then we obtain 
\begin{eqnarray}\label{t_2}
\sup _{L>0}\sup _{(\zeta _r)_r\in \mathcal {A}^L_t(\varphi )}
\E [|u(W_t,\varphi _t,S_t)-u(W'_t,\varphi '_t,\tilde{S}'_t)|] \longrightarrow  0 
\end{eqnarray}
as $(w',\varphi ',s')\rightarrow (w,\varphi ,s)$ 
by Lemma \ref {conti_u}. 
Now (\ref {t_1}) and (\ref {t_2}) imply 
\begin{eqnarray*}
\lim _{(w',\varphi ', s')\rightarrow (w,\varphi ,s)}\sup _{L>0}
(V^L_t(w,\varphi ,s ; u)-V^L_t(w',\varphi ',s' ; u))\leq 0. 
\end{eqnarray*}
A similar argument gives us 
\begin{eqnarray*}
\lim _{(w',\varphi ', s')\rightarrow (w,\varphi ,s)}\sup _{L>0}
(V^L_t(w',\varphi ',s' ; u)-V^L_t(w,\varphi ,s ; u))\leq 0. 
\end{eqnarray*}
This establishes the assertion. 
\end{proof}

\begin{lem}\label{conti_1} \ 
For any compact set $E\subset D$, 
\begin{eqnarray*}
\limsup _{r\uparrow t}\sup _{L>0}
\sup _{(w,\varphi ,s)\in E}(V^L_r(w,\varphi ,s ; u)-V^L_t(w,\varphi ,s ; u))\leq 0, \ \ 
t\in (0,1], \\
\limsup _{t\downarrow r}\sup _{L>0}
\sup _{(w,\varphi ,s)\in E}(V^L_r(w,\varphi ,s ; u)-V^L_t(w,\varphi ,s ; u))\leq 0, \ \ 
r\in [0,1). 
\end{eqnarray*} 
\end{lem}

\begin{proof} 
Let $r,t\in [0,1]$ with $r<t$. 
Lemmas \ref {conti_u} and \ref {cond_of_X1} imply 
\begin{eqnarray*}
\sup _{L>0}\sup _{(w,\varphi ,s)\in E}\sup _{(\zeta _v)_{v}\in \mathcal {A}^L_r(\varphi )}
\E [|u(W_r,\varphi _r,\exp (X_r))] - 
u(\tilde{W}_t,\tilde{\varphi }_t,\exp (\tilde{X}_t))|]\longrightarrow 0
\end{eqnarray*}
as $r\uparrow t$ and $t\downarrow r$, where 
$(W_v,\varphi _v,X_v)_{v} = 
\Xi ^X_r(w,\varphi ,s ; (\zeta _v)_v)$, 
$(\tilde{W}_v, \tilde{\varphi }_v, \tilde{X}_v)_{v} = 
\Xi ^X_t(w,\varphi ,s ; (\tilde{\zeta }_v)_v)$ 
and $\tilde{\zeta }_v = \zeta _v1_{[0,r]}(v)$ for 
$(\zeta _r)_r\in \mathcal {A}^L_t(\varphi )$. 
This implies the assertions.
\end{proof}

Similar arguments give us the following lemma: 

\begin{lem}\label{conti_2} \ For any $L>0$ and compact set $E\subset D$, 
\begin{eqnarray*}
\limsup _{r\uparrow t}\sup _{(w,\varphi ,s)\in E}
(V^L_t(w,\varphi ,s ; u)-V^L_r(w,\varphi ,s ; u))\leq 0, \ \ t\in (0,1], \\
\limsup _{t\downarrow r}\sup _{(w,\varphi ,s)\in E}
(V^L_t(w,\varphi ,s ; u)-V^L_r(w,\varphi ,s ; u))\leq 0, \ \ r\in [0,1). 
\end{eqnarray*} 
\end{lem}

By Lemmas \ref {conti_wphis_L}--\ref {conti_2}, we obtain Proposition \ref {conti_L}. 
We remark that Lemma \ref {nondecreasing_polygrowth} and Proposition \ref {conti_L} imply 
$V^L_t(\cdot ; u), V_t(\cdot ; u)\in \mathcal {C}$. 

We introduce a version of Theorem \ref {converge}, 
which will be used to prove Theorem \ref {semi} in the next section. 
Set 
\begin{eqnarray*}
\mathcal {A}^{n,L}_k(\varphi ) &=& \{ (\psi _l)^{k-1}_{l=0}\in \mathcal {A}^n_k(\varphi )\ ; \ 
\psi _l\leq L/n,\ l=0,\ldots , k-1 \}, \\
V^{n,L}_k(w,\varphi ,s ; u) &=& \sup _{(\psi _l)^{k-1}_{l=0}\in \mathcal {A}^{n,L}_k(\varphi )}
\E [u(W^n_k,\varphi ^n_k,S^n_k)]. 
\end{eqnarray*}
Note that $V^n_k(w,\varphi ,s ; u) = \sup _{L>0}V^{n,L}_k(w,\varphi ,s ; u)$. 
By similar arguments as in Section \ref {proof_of_converge}, we see that 
\begin{prop} \ \label{prop_converge_L}
For any $L>0$, $t\in [0,1]$ and $u\in \mathcal {C}$, the convergence 
\begin{eqnarray*}
\lim _{n\rightarrow \infty }V^{n,L}_{[nt]}(w,\varphi ,s;u) = V^L_t(w,\varphi ,s;u). 
\end{eqnarray*}
holds uniformly on any compact subset of $D$. 
\end{prop}

\subsection{Proof of Theorem \ref {semi}}\label{sec_properties}

We apply Nisio's method (\cite {Nisio}) to show Theorem \ref {semi}. \vspace{1mm}
We define the operators 
$Q^L_t : \mathcal {C}\longrightarrow \mathcal {C}$ and 
$Q^{n,L}_t : \mathcal {C}\longrightarrow \mathcal {C}$ by 
$Q^L_tu(w,\varphi ,s) = V^L_t(w,\varphi ,s ; u)$ and 
$Q^{n,L}_tu(w,\varphi ,s) = V^{2^n, L}_{[2^nt]}(w,\varphi ,s ; u)$. 
We see that $Q^L_t$ and $Q^{n,L}_t$ are well defined by 
the results in Section \ref {S-R} and 
the standard arguments of discrete-time dynamic programming theory 
(see \cite {Bertsekas-Shreve} for instance). 
First we show 
\begin{eqnarray}\label{semi_L}
Q^L_{t+r}u(w,\varphi ,s) = Q^L_tQ^L_ru(w,\varphi ,s)
\end{eqnarray}
for any $t, r\in I$ with $t+r\leq 1$, where 
$I = \{ k/2^l\ ; \ k,l\in \Bbb {Z}_+\}\cap [0,1]$. 
Let $n\in \Bbb {N}$ be large enough so that $2^nt, 2^nr\in \Bbb {Z}_+$. 
By the Bellman equation of the discrete-time case (\cite{Bertsekas-Shreve}), we have 
\begin{eqnarray}\label{temp1_0_I}
Q^{n,L}_{t+r}u(w,\varphi ,s) = 
Q^{n,L}_tQ^{n,L}_ru(w,\varphi ,s). 
\end{eqnarray}
By Proposition \ref {prop_converge_L}, 
we see that the left side of (\ref {temp1_0_I}) converges to that of (\ref {semi_L}) 
as $n\rightarrow \infty $ for any $t, r\in I$. 
The following proposition will give the convergence of the right side. 
\begin{prop} \ \label{prop_conv_rhs}
Let $u_n, u\in \mathcal {C}$ be utility functions satisfying $(\ref {growth_C})$ for some 
$C_u$ and $m_u$. Assume that $u_n$ converges to $u$ uniformly on any compact subset of $D$ as 
$n\rightarrow \infty $. Then
\begin{eqnarray*}
\lim _{n\rightarrow \infty }\sup _{k = 0, \ldots , n}
|V^{n,L}_k(w, \varphi , s ; u_n) - V^{n,L}_k(w, \varphi , s ; u)| = 0, \ \ 
(w, \varphi , s)\in D. 
\end{eqnarray*}
\end{prop}

\begin{proof} 
Take any $R > 0$. Then 
\begin{eqnarray*}
&&|V^{n,L}_k(w, \varphi , s ; u_n) - V^{n,L}_k(w, \varphi , s ; u)| \\
&&\leq 
\sup _{(w', \varphi ', s')\in D_R}|u_n(w', \varphi ', s') - u(w', \varphi ', s')| + 
\frac{C_0}{R}
\end{eqnarray*}
by Lemma \ref {cond_of_X2} and the Chebyshev inequality, 
where $C_0 > 0$ depends only on $b$, $\sigma $, $C_u$, $m_u$ and $(w, \varphi , s)$. 
Now we get the assertion by letting $n\rightarrow \infty $ and $R\rightarrow \infty $. 
\end{proof}

Using Proposition \ref {prop_conv_rhs} 
and the uniform convergence of $Q^{n, L}_ru$ to $Q^L_ru$ on any compact set, 
we see that the right side of (\ref {temp1_0_I}) converges to that of (\ref {semi_L}). 
Moreover, Proposition \ref {conti_L} implies that (\ref {semi_L}) also holds for any $t,r \in [0,1]$. 
\vspace{0.5mm}
Now Theorem \ref {semi} is obtained from (\ref {semi_L}), the relation 
$Q_tu(w,\varphi ,s) = \sup _{L>0}Q^L_tu(w,\varphi ,s)$, and 
a similar calculation to the proof of Proposition 4 in \cite{Nisio}.  \qed

\subsection{Proof of Theorem \ref {conti}}\label{proof_of_conti}

We now prove Theorem \ref {conti}. 
First we consider the right-continuity at $t=0$ when $h(\infty ) = \infty $.

\begin{lem} \ \label{eval_0}
Assume $h(\infty ) = \infty $. Then for any $t\in [0,1]$ and 
$(\zeta _r)_{0\leq r\leq t}\in \mathcal {A}_t(\varphi )$,
\begin{eqnarray}\label{temp_lemma_1}
\int ^r_0\exp \Big( -\int ^v_0g(\zeta _{v'})dv'\Big) \zeta _vdv \leq 
\phi (r), \ \ r\in [0,t], 
\end{eqnarray}
where $\phi (r),\ r\in (0,1]$, is a continuous function, depending only on function 
$h(\zeta )$ and $\Phi _0$, 
such that $\lim _{r\rightarrow 0}\phi (r) = 0$.  
\end{lem}

\begin{proof} 
Let $\pi _r = \int ^r_0g(\zeta _v)dv$ and 
$\tau _R = \inf \{ v\in [0,r]\ ; \ \pi _v>R \} \wedge r$ 
for $r\in (0,t]$ and $R>0$. Then we have 
\begin{eqnarray*}
\int ^r_0\exp ( -\pi _v) \zeta _vdv \leq  
\int ^{\tau _R}_0\zeta _vdv + 
\int ^r_{\tau _R}e^{-R}\zeta _vdv \leq 
\int ^{\tau _R}_0\zeta _vdv + \Phi _0e^{-R}
\end{eqnarray*}
for $r\in (0,t]$ and $R>0$. 
Since $g(\zeta )$ is convex, the Jensen inequality implies 
\begin{eqnarray*}
\int ^{\tau _R}_0\zeta _vdv \leq 
rg^{-1}\Big( \frac{1}{r}\int ^{r}_0g(\zeta _v1_{[0, \tau _R]})dv\Big) \leq 
rg^{-1}\Big( \frac{1}{r}\int ^{\tau _R}_0g(\zeta _v)dv\Big) \leq rg^{-1}(R/r), 
\end{eqnarray*}
where $g^{-1}(y) = \sup \{ \zeta \in [0,\infty ) \ ; \ g(\zeta ) = y \} $, $y\geq 0$. 
The function $g^{-1}(y)$ is well defined at any $y\geq 0$ 
and continuous for large $y$. 

If we can find a positive function $R(r)$ that satisfies 
\begin{eqnarray}\label{cond_Rr}
R(r) \longrightarrow \infty \ \ \mathrm {and} \ \ 
rg^{-1}(R(r)/r) \longrightarrow 0\ \ \mathrm {as} \ \ r\rightarrow 0, 
\end{eqnarray}
then we obtain (\ref {temp_lemma_1}) by letting $\phi (r) = rg^{-1}(R(r)/r) + \Phi _0\exp (-R(r))$. 
To construct such $R(r)$, 
let us define a function $f(\zeta )$, $\zeta \geq 0$, by $f(\zeta ) = 
\zeta \sqrt{h(\zeta /2)}$. Then $f(\zeta )$ is continuous, strictly increasing for large $y$ and 
satisfies $f(0)=0$ and $\lim _{\zeta \rightarrow\infty }f(\zeta ) = \infty $. 
Thus $f(\zeta )$ has an inverse function $f^{-1}(y)$ on $[0,\infty )$ such that 
$f^{-1}(0) = 0$, $\lim _{y\rightarrow \infty }f^{-1}(y) = \infty $, and $f^{-1}(y)$ is continuous for large $y$. 
So we can put $M(r) = f^{-1}(1/r)$ and $R(r) = rg(M(r))$ for $r\in (0,1]$. 
Then we see that $M(r), R(r) \longrightarrow \infty $ as $r\rightarrow 0$ and that 
\begin{eqnarray*}
R(r) \geq  r\int ^{M(r)}_{M(r)/2}h(\zeta )d\zeta  \geq  
\frac{rM(r)h(M(r)/2)}{2} = \frac{\sqrt{h(M(r)/2)}}{2} \longrightarrow \infty 
\end{eqnarray*}
as $r\rightarrow 0$. Moreover, we have 
\begin{eqnarray*}
rg^{-1}(R(r)/r) = rM(r) = \frac{1}{\sqrt{h(M(r)/2)}} \longrightarrow 0,\ \ r\rightarrow 0. 
\end{eqnarray*}
Then we obtain (\ref {cond_Rr}) and thus the assertion. 
\end{proof}
\begin{remark}
The above construction of $R(r)$ is somewhat artificial. 
Here, we give an image of the above proof. 
To make the situation simple, we consider only the case of 
$h(\zeta ) = C\zeta ^\alpha $ for some $C, \alpha > 0$. 
Then we have $g^{-1}(y) = \mathrm {Const.}\times \zeta ^{1/(1+\alpha )}$ and 
\begin{eqnarray*}
\phi (r) = \mathrm {Const.}\times \{ (r^\alpha R(r))^{1/(1+\alpha )} + \exp (-R(r))\} . 
\end{eqnarray*}
If we put $R(r) = r^{-\beta }$ with $\beta > 0$, then 
we observe 
\begin{eqnarray*}
\phi (r) = \mathrm {Const.}\times \{ r^{(\alpha  - \beta )/(1+\alpha )} + \exp (-r^{-\beta })\} , 
\end{eqnarray*}
which converges to $0$ as $r\rightarrow 0$ when $0 < \beta  < \alpha $. 
In the proof of Lemma \ref {eval_0}, 
$\beta $ was set as $\alpha / (2+\alpha )$. 

In the general case, the construction of $R(r)$ becomes a little complicated, 
and we need the auxiliary functions $f(\zeta )$ and $M(r)$. 
In the case of $h(\zeta ) = C\zeta ^\alpha $, they are represented as 
$f(\zeta ) = \mathrm {Const.}\times \zeta ^{1+\alpha /2}$ and 
$M(r) = \mathrm {Const.}\times r^{-2/(2+\alpha )}$. 

\end{remark}
\begin{prop} \ \label{th_conti_(i)_1}
Assume $h(\infty )=\infty $. Then for any compact set $E\subset D$, 
\begin{eqnarray*}
\lim _{t\downarrow 0}\sup _{(w,\varphi ,s)\in E}
|V_t(w,\varphi ,s ; u)-u(w,\varphi ,s)| = 0. 
\end{eqnarray*}
\end{prop}

\begin{proof}
Take any $t\in (0,1)$. 
Let $\hat{S}_t = s\exp \Big( -\int ^t_0g(\zeta _v)dv\Big) $ and 
$(W_r,\varphi _r,S_r)_{0\leq r\leq t} = \Xi _t(w,\varphi ,s ; (\zeta _r)_r)$. 
Then we have 
\begin{eqnarray}\label{temp15_1}
V_t(w,\varphi ,s ; u) - u(w,\varphi ,s) \leq  
\sup _{(\zeta _r)_r\in \mathcal {A}_t(\varphi )}
\left| \E [u(W_t,\varphi _t,S_t)] - 
\E [u(w,\varphi _t,\hat{S}_t)]\right| \hspace{5mm}
\end{eqnarray}
by the relations $\varphi _t\leq \varphi $ and $\hat{S_t}\leq s$. 
Using Lemma \ref {eval_0}, the Burkholder--Davis--Gundy inequality and the H\"older inequality, we have 
$\E [|S_t - \hat{S}_t|]\leq C_0st^{1/2}$ and 
\begin{eqnarray*}
\E [|W_t-w|] &\leq & 
s\E \Big[ \int ^t_0
\exp \Big (-\int ^r_0g(\zeta _v)dv\Big )\zeta _rdr\Big] + 
\E \Big[ \int ^t_0
| S_r - \hat{S}_r | \zeta _rdr \Big] \\
&\leq & 
s\phi (t) + C_0\Phi _0st^{1/2} 
\end{eqnarray*}
for some $C_0>0$ independent of $t, w, \varphi , s$ and $(\zeta _r)_r$. 
Then, by (\ref {temp15_1}) and Lemma \ref {conti_u}, we get 
$\limsup _{t\downarrow 0}\sup _{(w,\varphi ,s)\in E}
(V_t(w,\varphi ,s ; u)-u(w,\varphi ,s)) \leq  0$. 
The inequality $\limsup _{t\downarrow 0}\sup _{(w,\varphi ,s)\in E}
(u(w,\varphi ,s)-V_t(w,\varphi ,s ; u)) \leq  0$ is 
obtained by Lemma \ref {conti_1}. This yields the assertion. 
\end{proof}

Next we consider the case of $h(\infty )<\infty $. 

\begin{prop} \ \label{prop_conti_3_1}Assume $h(\infty )<\infty $. 
Then for any compact set $E\subset D$ 
\begin{eqnarray*}
\limsup _{t\downarrow 0}\sup _{(w,\varphi ,s)\in E}
(V_t(w,\varphi ,s ; u) - Ju(w,\varphi ,s))\leq 0. 
\end{eqnarray*}
\end{prop}

\begin{proof} 
Fix $t\in (0,1)$ and $(\zeta_r)_{0\leq r\leq t}\in \mathcal {A}_t(\varphi )$. 
Let $(W_r, \varphi _r, X_r)_{0\leq r\leq t} = 
\Xi ^X_t(w,\varphi ,s ; \allowbreak (\zeta _r)_r)$. We easily have 
\begin{eqnarray}\nonumber 
&&\lim _{t\downarrow 0}\sup _{(w,\varphi ,s)\in E}
\sup _{(\zeta _r)_r\in \mathcal{A}_t(\varphi )}
\Big| \E [u(W_t, \varphi _t, S_t)] \\\label{temp_T}&&\hspace{30mm} - 
\E \Big [u\Big( w+s\int ^t_0e^{-\tilde{\eta }_r}\zeta _rdr, 
\varphi -\eta _t, se^{-\tilde{\eta }_t}\Big) \Big ]\Big| = 0 \hspace{5mm}
\end{eqnarray}
by Lemma \ref {conti_u}, \vspace{2mm}
where $\eta _r = \int ^r_0\zeta _vdv$ and $\tilde{\eta }_r = \int ^r_0g(\zeta _v)dv$. 
Now we define 
\begin{eqnarray*}
\hat{\eta }_r = 1_{(0,t]}(r)\int ^{\eta _r}_0h(\zeta '/r)d\zeta ', \ \ 
\hat{w}_t = \int ^{\eta _t}_0
\exp \left( -\int ^p_0h(\zeta '/t)d\zeta '\right) dp. 
\end{eqnarray*}
Since $g(\zeta )$ is convex, the Jensen inequality implies 
$\tilde{\eta }_r\geq rg(\eta _r/r) = \hat{\eta }_r$ and 
\begin{eqnarray*}
\hat{w}_t \geq \int ^t_0\exp \left( -\int ^{\eta _r}_0h(\zeta '/r)d\zeta '\right) \zeta _rdr \geq  
\int ^t_0e^{-\tilde{\eta }_r}\zeta _rdr
\end{eqnarray*}
for $r\in (0,t]$. 
Moreover, $h(\zeta )$ is non-decreasing in $\zeta $ and so is 
$u(w,\varphi ,s)$ in $w$. Thus we get 
\begin{eqnarray}\nonumber \label{temp_T4}
\E \Big [u\Big( w+s\int ^t_0e^{-\tilde{\eta }_r}\zeta _rdr, 
\varphi -\eta _t, se^{-\tilde{\eta }_t}\Big) \Big ]\leq 
\E [u( w+s\hat{w}_t, \varphi -\eta _t, se^{-\hat{\eta }_t})]
\end{eqnarray}
for any $(\zeta _r)_r\in \mathcal {A}_t(\varphi )$. 
By this inequality and (\ref {temp_T}), we get 
\begin{eqnarray}\nonumber \label{temp_T_2}
&&\limsup _{t\downarrow 0}\sup _{(w,\varphi ,s)\in E}\big (V_t(w,\varphi ,s ; u)\\
&&\hspace{25mm} - 
\sup _{(\zeta _r)_r\in \mathcal{A}_t(\varphi )}
\E [u( w+s\hat{w}_t, \varphi -\eta _t, se^{-\hat{\eta }_t})]\big )\leq 0. 
\end{eqnarray}

Next let us define 
\begin{eqnarray}\label{def_epsilon_F}
\tilde{\varepsilon }_t = 
\int ^{\Phi _0}_0(h(\infty )-h(\zeta /t))d\zeta , \ \ 
F(\psi ) = \int ^\psi _0e^{-h(\infty )p}dp. 
\end{eqnarray}
Then we have $|e^{-\hat{\eta }_t}-e^{-h(\infty )\eta _t}|\leq 4\tilde{\varepsilon }_t$ and 
$|\hat{w}_t-F(\eta _t)|\leq 4\Phi _0\tilde{\varepsilon }_t$. 
Since the dominated convergence theorem implies 
$\tilde{\varepsilon }_t\longrightarrow 0$ as $t\downarrow 0$, 
Lemma \ref {conti_u} then gives us 
\begin{eqnarray*}&&
\lim _{t\downarrow 0}\sup _{(w,\varphi ,s)\in E}
\sup _{(\zeta _r)_r\in \mathcal{A}_t(\varphi )}\big| 
\E [u( w+s\hat{w}_t, \varphi -\eta _t, s\exp ( -\hat{\eta }_t))]\\&&\hspace{38mm} - 
\E [u( w+F(\eta _t)s, \varphi -\eta _t, se^{-h(\infty )\eta _t})] \big| = 0. 
\end{eqnarray*}
By this and (\ref {temp_T_2}), we get the assertion. 
\end{proof}

\begin{prop} \ \label{prop_conti_3_2}Assume $h(\infty )<\infty $. 
Then for any compact set $E\subset D$, 
\begin{eqnarray*}
\limsup _{t\downarrow 0}\sup _{(w,\varphi ,s)\in E}
(Ju(w,\varphi ,s) - V_t(w,\varphi ,s ; u))\leq 0. 
\end{eqnarray*}
\end{prop}

\begin{proof} 
Suppose $t\in (0,1)$. For any $(w,\varphi ,s)\in E$, fix a $\psi \in [0,\varphi ]$ and define 
$(\zeta _r)_{0\leq r\leq t}\in \mathcal {A}_t(\varphi )$ by 
$\zeta _r = \psi / t$ and $(W_r,\varphi _r,S_r)_{0\leq r\leq t} = 
\Xi _t(w,\varphi ,s ; (\zeta _r)_r)$. Similarly to the proof of 
Proposition \ref {prop_conti_3_1}, we get 
\begin{eqnarray*}
\lim _{t\downarrow 0}\sup _{(w,\varphi ,s)\in E}\sup _{\psi \in [0,\varphi ]}
\big| u(w+F(\psi )s,\varphi -\psi ,se^{-h(\infty )\psi }) - 
\E [u(W_t,\varphi _t,S_t)]\big| = 0, 
\end{eqnarray*}
which implies our assertion. 
\end{proof}

Finally, we consider the continuity with respect to $t\in (0,1]$. 

\begin{prop} \ \label{prop_conti(i)_2}
For any compact set $E\subset D$, \\
$\mathrm {(i)}$ \ $\lim _{t'\uparrow t}\sup _{(w,\varphi ,s)\in E}
|V_{t'}(w,\varphi ,s ; u)-V_t(w,\varphi ,s ; u)| = 0$, \ \ $t\in (0,1]$, \\
$\mathrm {(ii)}$ \ $\lim _{t'\downarrow t}\sup _{(w,\varphi ,s)\in E}
|V_{t'}(w,\varphi ,s ; u)-V_t(w,\varphi ,s ; u)| = 0$, \ \ $t\in (0,1)$. 
\end{prop} 

\begin{proof} 
Lemma \ref {conti_1} implies 
\begin{eqnarray*}
\limsup _{t'\uparrow t}\sup _{(w,\varphi ,s)\in E}
(V_{t'}(w,\varphi ,s ; u)-V_t(w,\varphi ,s ; u)) \leq 0. 
\end{eqnarray*}
By the following uniform convergence (which is given by Dini's theorem) 
\begin{eqnarray*}
\lim _{L\rightarrow \infty }\sup _{(w, \varphi , s)\in E}|V^L_t(w,\varphi ,s ; u) - V_t(w,\varphi ,s ; u)| = 0
\end{eqnarray*}
and Lemma \ref {conti_2}, we have 
\begin{eqnarray*}
\limsup _{t'\uparrow t}\sup _{(w,\varphi ,s)\in E}
(V_t(w,\varphi ,s ; u)-V_{t'}(w,\varphi ,s ; u)) \leq 0. 
\end{eqnarray*}
This gives assertion (i). 

Next we check (ii). 
If $h(\infty )=\infty $, this assertion holds by Proposition \ref {th_conti_(i)_1} and 
Theorem \ref {semi}, so we may assume $h(\infty )<\infty $. 

By Propositions \ref {prop_conti_3_1}--\ref {prop_conti_3_2} and Theorem \ref {semi}, we get 
\begin{eqnarray*}
\lim _{t'\downarrow t}\sup _{(w,\varphi ,s)\in E}|V_{t'}(w,\varphi ,s ; u) - 
JV_t(w,\varphi ,s ; u)| = 0, 
\end{eqnarray*}
and obviously
$V_t(w,\varphi ,s ; u)\leq JV_t(w, \varphi , s ; u)$. 
So, it suffices to show 
\begin{eqnarray}\label{for_prove}
JV_t(w,\varphi ,s ; u) \leq  V_t(w,\varphi ,s ; u), \ \ t > 0. 
\end{eqnarray}

Fix a $\psi \in [0,\varphi ]$ and 
a $(\zeta _r)_{0\leq r\leq t}\in \mathcal {A}_t(\varphi -\psi )$. 
Let $\delta \in (0,t)$ and 
define $(\tilde {\zeta }_r)_{0\leq r\leq t}\in \mathcal {A}_t(\varphi )$ by 
$\tilde{\zeta }_r = (\psi/\delta )1_{[0,\delta ]}(r)+\zeta _r$. 
Put $(W_r,\varphi _r,X_r)_{0\leq r\leq t} = 
\Xi ^X_t( w+F(\psi )s,\varphi -\psi ,se^{-h(\infty )\psi } ; (\zeta _r)_r) $ and 
$(\tilde{W}_r, \tilde{\varphi }_r, \tilde{X}_r)_{0\leq r\leq t} = 
\Xi ^X_t(w,\varphi ,s ; (\tilde{\zeta }_r)_r)$, where $F(\psi )$ is given by (\ref {def_epsilon_F}). 
Then we have for $r\in [\delta ,t]$ 
\begin{eqnarray*}
\tilde{X}_r - X_r = 
\int ^r_0(\sigma (\tilde{X}_v)-\sigma (X_v))dB_v + 
\int ^r_0(b(\tilde{X}_v)-b(X_v))dv + e_\delta , 
\end{eqnarray*}
where 
\begin{eqnarray*}
e_\delta  = h(\infty )\psi - \int ^\delta _0(g(\tilde{\zeta _v}) - g(\zeta _v))dv = 
\frac{1}{\delta }\int ^\delta _0\int ^\psi _0
\left(h(\infty )-h\left (\frac{\zeta '}{\delta } + \zeta _v\right )\right )d\zeta 'dv. 
\end{eqnarray*}
Using the Burkholder--Davis--Gundy inequality and the H\"older inequality, 
we get 
\begin{eqnarray*}&&
\E [\sup _{v\in [\delta ,r]}|\tilde{X}_v-X_v|^2] 
\leq 
C_0\Big\{ \int ^r_\delta 
\E [\sup _{v'\in [\delta ,v]}|\tilde{X}_{v'}-X_{v'}|^2]dv + 
\delta  + \E [e_\delta ]\Big\} , \ \ 
r\in [\delta ,t]
\end{eqnarray*}
for some $C_0 > 0$ depending only on $b$, $\sigma $ and $E$. 
Since $\E [e_\delta ]\leq \tilde{\varepsilon }_\delta \longrightarrow 0$ 
as $\delta \rightarrow 0$, where $\tilde{\varepsilon }_\delta $ is given by (\ref {def_epsilon_F}), we get 
$\E [\sup _{r\in [\delta ,t]}|\tilde{X}_r-X_r|^2]\longrightarrow 0$: 
therefore, 
$\E [\sup _{r\in [0,t]}|\exp (\tilde{X}_r)-\exp (X_r)|] \longrightarrow 0$ 
as $\delta \rightarrow 0$ 
by the above inequality and the Gronwall inequality. 
Moreover, by these convergences, the boundedness of $(\zeta _r(\omega ))_{r, \omega }$, and Lemma \ref {cond_of_X1}, 
we can show the convergence 
$\E [|\tilde {W}_t-W_t|] \longrightarrow 0$ as $\delta \rightarrow 0$. 
Now we can apply Lemma \ref {conti_u} to obtain 
\begin{eqnarray}\label{conv_JV}
\lim _{\delta \rightarrow 0}
\big |\E [u(W_t,\varphi _t,\exp (X_t))] - 
\E [u(\tilde{W}_t,\tilde{\varphi }_t,\exp (\tilde{X}_t))]\big | = 0. 
\end{eqnarray}
By (\ref {conv_JV}), we easily get $\E [(W_t, \varphi _t, \exp (X_t))] \leq V_t(w, \varphi , s ; u)$. 
Since $(\zeta _r)_{r}\in \mathcal {A}_t(\varphi -\psi )$ was arbitrary, 
and $\psi \in [0,\varphi ]$ was also arbitrary, we get (\ref {for_prove}). 
\end{proof}

Using Propositions \ref {th_conti_(i)_1}--\ref {prop_conti(i)_2} and 
$V_t(\cdot \ ; u) \in \mathcal {C}$, 
we complete the proof of Theorem \ref {conti}.

\subsection{Proof of Proposition \ref {prop_suff_cond}}\label{proof_suff_cond}

Fix $t\in (0, 1]$ and $(w, \varphi , s)\in \hat{U}$. 
First, we will show that $V_t(w, \varphi , s ; u) > U(w)$. 
Define $(\bar{\zeta }_r)_r\in \mathcal {A}_t(\varphi )$ by $\bar{\zeta }_r = \varphi /t$, $r\in [0, t]$ and 
let $(\bar{W}_r, \bar{\varphi }_r, \bar{S}_r)_r = \Xi _t(w, \varphi , s ; (\bar{\zeta }_r)_r)$. 
Then, by the definition of $V_t(w, \varphi , s ; u)$, the boundedness of $b$ and $\sigma $, [C1], 
and the Jensen inequality, 
we can easily observe 
\begin{eqnarray*}
V_t(w, \varphi , s ; u) - U(w) \geq  
\frac{\delta \varphi }{t}\int ^t_0\E [\bar{S}_r]dr 
\geq 
\delta \varphi se^{-(K + g(\varphi /t))t} > 0 
\end{eqnarray*}
for some $K > 0$. 
Here we denote $\hat{\delta } = V_t(w, \varphi , s ; u) - U(w) > 0$ for brevity. 

Next, fix any $\varepsilon \in (0, 1)$ and 
$\eta \in (0, \hat{\delta }/2)$. 
Then there exists $(\zeta _r)_r\in \mathcal {A}_t(\varphi )$ such that 
\begin{eqnarray}\label{eps_optimal}
V_t(w, \varphi , s ; u) < \E [U(W_t)] + \eta , 
\end{eqnarray}
where $(W_r, \varphi _r, S_r)_r = \Xi _t(w, \varphi , s ; (\zeta _r)_r)$ 
(which depends on $\eta $, whereas is independent of $\varepsilon $). 
Put 
$(\tilde{W}_r, \tilde{\varphi }_r, \tilde{S}_r)_r = \Xi _t(w, \varphi , s + \varepsilon  ; (\zeta _r)_r)$ 
(note that $\tilde {W}_r \geq W_r$, $\tilde{S}_r \geq S_r$ and $\tilde{\varphi }_r = \varphi _r$.) 
Then we have 
\begin{eqnarray}
V_t(w, \varphi , s + \varepsilon  ; u) - V_t(w, \varphi  , s ; u) 
\geq 
\varepsilon \delta \E\left [\int ^t_0\zeta _rA^\varepsilon _rdr\right ] - \eta 
\label{diff_V_ineq}
\end{eqnarray}
by [C1] and (\ref {eps_optimal}), where $A^\varepsilon _r = (\tilde{S}_r - S_r)/\varepsilon $. 
Here, we will show that $(A^\varepsilon _r)_r$ converges to 
a process $(A_r)_r$ in the following sense: 
\begin{eqnarray}\label{conv_Ar}
\E \left [\int ^t_0\zeta _r|A^\varepsilon _r - A_r|dr\right ] \ \longrightarrow \ 0, \ \ \varepsilon \rightarrow 0, 
\end{eqnarray}
and $(A_r)_r$ is given by $A_r = S_rL_r/s$, where 
$(L_r)_r$ is the solution of the SDE
\begin{eqnarray*}
\left\{ 
\begin{array}{l}
dL_r = b'(X_r)L_rdr + \sigma '(X_r)L_rdB_r, \ \ r > 0, \\
\hspace{1.9mm}L_0 = 1. 
\end{array}
\right. 
\end{eqnarray*}
Note that existence and uniqueness of the above SDE are guaranteed by [C2]. 
Moreover, Ito's formula implies that $L_r = \exp (\Lambda _r) > 0$, where 
\begin{eqnarray*}
\Lambda _r = \int ^r_0\left\{ b'(X_v) - \frac{1}{2}\sigma '(X_v)\right\} dv + 
\int ^r_0\sigma '(X_v)dB_v. 
\end{eqnarray*}

Define $X_r = \log S_r$, $\tilde{X}_r = \log \tilde{S}_r$ and 
$L^\varepsilon _r = s(\tilde{X}_r - X_r) / \varepsilon $. 
By using $S_rL^\varepsilon _r \leq sA^\varepsilon \leq \tilde{S}_rL^\varepsilon _r$ and 
$\tilde{S}_r \leq \hat{Z}(s + 1)$, we can get 
\begin{eqnarray}
\nonumber 
\E \left [\int ^t_0\zeta _r|A^\varepsilon _r - A_r|dr\right ] 
&\leq & 
\frac{\varphi }{s}\Bigg\{ \E [\hat{Z}(s + 1)^2]^{1/2}
\E \left[ \sup _{0\leq r\leq t}\left | L^\varepsilon _r - L_r \right |^2\right ]^{1/2}\\
&& + 
\E \left[ \sup _{0\leq r\leq t}|\tilde{S}_r - S_r|^2\right ] ^{1/2}
\E \left[ \sup _{0\leq r\leq t}|L_r|^2\right ] ^{1/2}\Bigg \} . \hspace{8mm}
\label{temp_ineq_Step3}
\end{eqnarray}

Now we consider the limit of the right side of (\ref {temp_ineq_Step3}) as $\varepsilon \rightarrow 0$. 
By [C2], Lemmas \ref {cond_of_X2}, \ref {cond_of_X3}, Theorem 2.5.9 in \cite {Krylov} and the inequality
\begin{eqnarray}\label{exp_log}
|e^x-e^y|\leq \int ^1_0e^{rx}e^{(1-r)y}dv|x-y|\leq (e^x+1)(e^y+1)|x-y|, 
\end{eqnarray}
we have 
\begin{eqnarray}\label{est_diff_S}
\E \left[ \sup _{0\leq r\leq t}|\tilde{S}_r - S_r|^2\right ]  \leq 
4\E [\hat{Z}(s + 1)^4]^{1/2}\E [\sup _{0\leq r\leq t}|\tilde{X}_r - X_r|^4]^{1/2} \leq 
K'\varepsilon ^2
\end{eqnarray}
for some $K' > 0$ depending only on $s$, $b$ and $\sigma $. 

Here, we denote by $X(\cdot ; x_0 , (\tilde{\zeta }_r)_r)$ the solution of (\ref {SDE_X}), 
given $(\tilde{\zeta }_r)_r\in \mathcal {A}_t(\varphi )$. 
Then, similarly to Theorem 4.6.5 in \cite {Kunita}, 
by [C2] and Theorem 2.5.9 in \cite {Krylov}, 
we can show 
that 
the process $(\partial X/\partial x)(\cdot  ; x, \allowbreak (\tilde{\zeta }_r)_r)$ exists for each $x\in \Bbb {R}$, 
that $(\partial X/\partial x)(r ; \log s, (\zeta _r)_r) = L_r$, and that 
the following convergence holds for each $x$: 
\begin{eqnarray}\nonumber 
&&\sup _{(\tilde{\zeta }_r)_r\in \mathcal {A}_t(\varphi )}
\E \left [ \sup _{0\leq r\leq 1}
\left | \frac{X(r ; x + \varepsilon , (\tilde {\zeta }_r)_r) - X(r ; x, (\tilde {\zeta }_r)_r)}{\varepsilon } - 
\frac{\partial }{\partial x}X(r ; x, (\tilde {\zeta }_r)_r) \right |^2\right ] \\
&&\longrightarrow \ 0, \ \ \varepsilon \rightarrow 0. \label{conv_flow_2}
\end{eqnarray}
By (\ref {conv_flow_2}) and a standard calculation, we can show that 
$\E \left[ \sup _{0\leq r\leq 1}\left | L^\varepsilon _r - L_r \right |^2\right ] $ 
converges to zero as $\varepsilon \rightarrow 0$. 
By combining this with (\ref {temp_ineq_Step3}) and (\ref {est_diff_S}), 
we obtain (\ref {conv_Ar}). 
Here, we stress that, by using (\ref {conv_flow_2}) and Theorem 2.5.9 in \cite {Krylov} again, 
we can generalise (\ref {conv_Ar}) to the following sharper estimation: 
\begin{eqnarray}\label{unif_conv_A}
c(\varepsilon ) = \sup _{(\tilde{\zeta }_r)_r\in \mathcal {A}_t(\varphi )}
\E \left [\int ^t_0\tilde{\zeta }_r|A^\varepsilon _r(\tilde {\zeta }) - A_r(\tilde {\zeta })|dr\right ] \ \longrightarrow \ 0, \ \ \varepsilon \rightarrow 0, 
\end{eqnarray}
where 
$A^\varepsilon _r(\tilde{\zeta })$ and $A_r(\tilde{\zeta })$
are defined for each $(\tilde{\zeta }_r)_r$ 
in a way similar to the definitions of 
$A^\varepsilon _r$ and $A_r$. 
We omit a detailed proof of (\ref {unif_conv_A}). 

Next, set $\tilde{\delta } = \tilde{\delta }(\eta ) = \E [W_t] - w$ 
(recall that $W_t$ denotes the cash holdings satisfying (\ref {eps_optimal})). 
By [C1], the function $U$ has the inverse function $U^{-1}$ which is also continuously differentiable and 
$(U^{-1})'(y) = 1/U'(U^{-1}(y))\in (0, 1/\delta )$, $y\in \{ U(w)\ ; \ w\in \Bbb {R} \}$. 
Note that $U^{-1}$ is strictly increasing and $(U^{-1})'$ is non-decreasing because $U$ is concave. 
Here, applying the Jensen inequality, we have 
\begin{eqnarray}\nonumber 
\tilde {\delta } 
&\geq & 
U^{-1}(\E [U(W_t)]) - U^{-1}(U(w))
> 
U^{-1}(V_t(w, \varphi , s ; u) - \eta ) - U^{-1}(U(w))\\
&=& 
\int ^1_0(U^{-1})'\left( U(w) + k(\hat{\delta } - \eta )\right) dk(\hat{\delta } - \eta )  \ > \ \bar{\delta }, 
\label{temp_bar_delta}
\end{eqnarray}
where $\hat{\delta }$ is defined at the end of Step 1 and 
$\bar{\delta } = \hat{\delta } / (2U'(w)) > 0$. Note that $\bar{\delta }$ is independent of $\eta $. 
By (\ref {diff_V_ineq}), (\ref {unif_conv_A}) and (\ref {temp_bar_delta}), 
we get 
\begin{eqnarray}\nonumber 
\frac{V_t(w, \varphi , s + \varepsilon  ; u) - V_t(w, \varphi  , s ; u)}{\varepsilon } 
\geq 
-\delta c(\varepsilon ) + 
\frac{\delta \tilde{\delta }}{s}\int _{[0, T]\times \Omega }e^{\Lambda _r(\omega )}\nu (dr, d\omega ) 
 - \frac{\eta }{\varepsilon }, \\
\label{temp_exp1}
\end{eqnarray}
where 
$\nu (dr, d\omega ) = \tilde{\delta }^{-1}\zeta _r(\omega )S_r(\omega )drP(d\omega )$
is the probability measure on $([0, t]\times \Omega , \mathcal {B}([0, t])\otimes \mathcal {F})$. 
We can apply the Jensen inequality to obtain 
\begin{eqnarray}
\int _{[0, t]\times \Omega }e^{\Lambda _r(\omega )}\nu (dr, d\omega ) 
\geq 
\exp \left( \frac{1}{\tilde{\delta }}
\E \left[ \int ^t_0\zeta _rS_r\Lambda _rdr \right] \right) . \label{temp_exp2}
\end{eqnarray}
Using [C2], the relation $(\zeta _r)_r\in \mathcal {A}_t(\varphi )$, 
Lemma \ref {cond_of_X2}, the definition of $(\Lambda _r)_r$, 
the H\"older inequality and the Burkholder--Davis--Gundy inequality, we have 
\begin{eqnarray}
\left| \E \left[ \int ^t_0\zeta _rS_r\Lambda _rdr \right] \right| 
\leq  
\varphi \E \left [\sup _{0\leq r\leq t}|\Lambda _r|^2\right ]^{1/2}\E [\hat{Z}(s)^2]^{1/2}
\leq K'' \varphi 
\label{temp_exp3}
\end{eqnarray}
for some 
$K'' > 0$ which depends only on $t$, $s$, $b'$ and $\sigma '$. 
By 
(\ref {temp_bar_delta})--(\ref {temp_exp3}), we get 
\begin{eqnarray*}
\frac{V_t(w, \varphi , s + \varepsilon  ; u) - V_t(w, \varphi  , s ; u)}{\varepsilon } 
\geq 
-\delta c(\varepsilon ) + \frac{\delta \bar{\delta }}{s}e^{-K''\varphi / \bar{\delta }} - \frac{\eta }{\varepsilon }. 
\end{eqnarray*}
Letting $\eta \rightarrow 0$ and then taking $\liminf _{\varepsilon \rightarrow 0}$, 
we see from (\ref {unif_conv_A}) that the left side of (\ref {cond_diff}) has the lower bound 
$\delta \bar{\delta }e^{ - K''\varphi / \bar{\delta }}/s > 0$. 
This completes the proof. 
\qed

\subsection{Proof of Theorem \ref {HJB_th}}\label{sec_HJB}

In Sections $\ref {sec_HJB}$ and $\ref {sec_uniqueness}$ we assume that $h$ is strictly increasing 
and $h(\infty ) = \infty $. 
First we consider the characterisation of $V^L_t(w, \varphi , s ; u)$ 
as the viscosity solution of the corresponding HJB. 
We define a function $F^L : \mathscr {S}\longrightarrow \Bbb {R}$ by 
\begin{eqnarray*}
F^L(z, p, X) = -\sup _{0\leq \zeta \leq L}
\left\{ \frac{1}{2}\hat{\sigma }(z_s)^2X_{ss} + \hat{b}(z_s)p_s + 
\zeta \left( z_sp_w - p_\varphi \right) - g(\zeta )z_sp_s\right\} . 
\end{eqnarray*}

\begin{prop}\label{HJB_L_th}
Assume $h(\infty ) = \infty $. Then, for any $u\in \mathcal {C}$, 
the function $V^L_t(w, \varphi ,\allowbreak s ; u)$ is the viscosity solution of 
\begin{eqnarray}\label{HJB2_L}
\frac{\partial }{\partial t}v + F^L(z, \mathcal {D}v, \mathcal {D}^2v) = 0 \ \ 
\mathrm {on} \ (0, 1]\times \hat{U}. 
\end{eqnarray}
\end{prop}

Since the control region $[0, L]$ is compact, we obtain 
Proposition \ref {HJB_L_th} using (\ref {semi_L}) and the standard arguments of 
the Bellman principle and HJB (see Theorem 5.4.1 in \cite{Nagai}). 

Next we treat HJB (\ref {HJB2}). 
Let $\mathscr {U} = \{ (z, p, X)\in \mathscr {S}\ ; \ F(z, p, X) > -\infty  \} $. 
A direct calculation proves the next proposition. 
\begin{prop} \ \label{prop_F_conti} For $(z, p, X)\in \mathscr {U}$,
\begin{eqnarray*}
F(z, p, X) &=& - \frac{1}{2}\hat{\sigma }(z_s)^2X_{ss} - \hat{b}(z_s)p_s\\&& - 
\max \left \{ \zeta ^*(z, p)\left( z_sp_w - p_\varphi \right) - g(\zeta ^*(z, p))z_sp_s, 0 \right \} , 
\end{eqnarray*}
where $\zeta ^*(z, p) = h^{-1}\left( \frac{z_sp_w - p_\varphi }{z_sp_s}\vee h(0)\right) 1_{\{ p_s > 0 \} }$. 
In particular, $F$ is continuous on $\mathscr {U}$. 
\end{prop}

Now we prove Theorem \ref {HJB_th}. 
We define an open set $\mathscr {R} = \hat{U}\times (\Bbb {R}^2\times (0, \infty ))\times S^3 \subset \mathscr {U}$. 
Since $F$ is continuous on $\mathscr {R}$ and $F^L$ converges to $F$ monotonically, 
we see that this convergence is uniform on any compact set in $\mathscr {R}$, by Dini's theorem. 
Similarly, using Dini's theorem again, we see that $V^L$ 
converges to $V$ uniformly on any compact set in $[0, 1]\times \hat{D}$. 
Moreover, we note that if we take $\hat{v}\in C^{1, 2}((0, 1]\times \hat{U})$ such that 
$V - \hat{v}$ has $0$ as a local maximum at $(t, z)$, then (\ref {cond_diff}) implies 
$(\partial \hat{v} / \partial z_s)(t, z) > 0$ and 
$(z, \mathcal {D}\hat{v}(t, z), \mathcal {D}^2\hat{v}(t, z))\in \mathscr {R}$. 
Then the same arguments as in the proof of Lemma 5.7.1 in \cite {Nagai} lead us to the assertion.  \qed

\subsection{Proof of Theorem \ref {unique_th}}\label{sec_uniqueness}

First we remark that Lemma \ref {nondecreasing_polygrowth} implies that 
$V_t(w, \varphi , s ; u)$ grows polynomially in $w, \varphi $ and $s$. 

Let $\tilde{U}\subset \hat{U}$ be open and bounded. 
Let $\mathscr {P}^{2, \pm }_{(0, 1]\times \tilde{U}}$ be parabolic variants of semijets and 
$\overline {\mathscr {P}}^{2, \pm }_{(0, 1]\times \tilde{U}}$ be their closures 
(see \cite {Crandall-Lions-Ishii}). 
For any $\lambda > 0$, we define $F_\lambda (z, r, p, X) = \lambda r + F(z, p, X)$. 
We see that the following are equivalent. 
\begin{description}
 \item[(a.)] \ A function $v$ is a viscosity subsolution (resp. supersolution) of (\ref {HJB2}), 
 \item[(b.)] \ A function $v_\lambda (t, z) = e^{-\lambda t}v(t, z)$ is a viscosity subsolution (resp. supersolution) of 
\begin{eqnarray}\label{HJB3}
\frac{\partial }{\partial t}v + F_\lambda (z, v, \mathcal {D}v, \mathcal {D}^2v) = 0. 
\end{eqnarray}
\end{description}

The same proof as Proposition 2.6 in \cite{Koike} gives the following lemma: 
\begin{lem} \label{semijet}
Suppose $v$ is a viscosity subsolution $($resp., supersolution$)$ of $(\ref {HJB3})$. 
Then
\begin{eqnarray*}
a + F_\lambda (z, v(t, z), p, X)\leq 0 \ ({\it resp.}, \ \geq 0)
\end{eqnarray*}
for any $(t, a, z, p, X)\in (0, 1]\times \Bbb {R}\times \tilde{U}\times \Bbb {R}^3\times S^3$ with 
$(a, p, X)\in \overline {\mathscr {P}}^{2,+}_{(0, 1]\times \tilde{U}}v(t, z)$ 
$($resp., $(a, p, X)\in \overline {\mathscr {P}}^{2,-}_{(0, 1]\times \tilde{U}}v(t, z) )$. 
\end{lem}

In particular, we note that 
\begin{eqnarray}\nonumber 
\overline {\mathscr {P}}^{2,-}_{(0, 1]\times \tilde{U}}v(t, z)
&\subset & 
\Bbb {R}\times \{(p, X)\ ; \ F_\lambda (z, v(t, z), p, X) > -\infty \} \\
&=& 
\Bbb {R}\times \{(p, X)\ ; \ F(z, p, X) > -\infty \} \label{temp_structure}
\end{eqnarray}
when $v$ is a viscosity supersolution of (\ref {HJB3}). 
Now we consider the comparison principle on a bounded domain. 

\begin{prop} \ \label{comparison_bdd} 
Suppose $v$ $($resp., $v' )$ is a viscosity subsolution $($resp., \allowbreak super\-solution$)$ of $(\ref {HJB3})$ on 
$(0, 1]\times \tilde{U}$. Moreover suppose 
$v(0, z) \leq v'(0, z)$ for $z\in \tilde{U}$ and 
$v \leq 0 \leq v'$ on $(0,1]\times \partial \tilde{U}$. 
Then $v \leq v'$ on $[0, 1]\times \tilde{U}$. 
\end{prop}

By 
(\ref {temp_structure}) and Theorem 8.12 in \cite{Crandall-Lions-Ishii}, 
we see that to prove Proposition \ref {comparison_bdd} 
it suffices to show the following Proposition \ref{structure_prop}.

\begin{prop} \ \label{structure_prop} 
The function $F_\lambda $ satisfies 
\begin{eqnarray*}
F_\lambda (z', r, \alpha (z-z'), Y) - F_\lambda (z, r, \alpha (z-z'), X) \leq 
\rho \left ( \alpha |z-z'|^2 + |z-z'| \right )
\end{eqnarray*}
for $\lambda > 0$, $\alpha > 1$, $r\in \Bbb {R}$, 
$z, z'\in \tilde{U}$, $X, Y\in S^3$ with $F(z', \alpha (z-z'), Y) > -\infty $ and 
\begin{eqnarray}\label{cond_XY}
-3\alpha 
\left(
\begin{array}{cc}
 I	& O	\\
 O	& I
\end{array}
\right)
\leq 
\left(
\begin{array}{cc}
 X	& O	\\
 O	& -Y
\end{array}
\right)
\leq 
3\alpha 
\left(
\begin{array}{cc}
 I	& -I	\\
 -I	& I
\end{array}
\right) , 
\end{eqnarray}
where $I\in \Bbb {R}^3\otimes \Bbb {R}^3$ denotes the unit matrix and 
$\rho : [0, \infty ) \longrightarrow [0, \infty )$ is a continuous function with $\rho (0) = 0$. 
\end{prop}

\begin{proof} 
Note that $F(z', \alpha (z-z'), Y) > -\infty $ implies 
$(z', \alpha (z-z'), Y)\in \mathscr {U}$, and thus 
\noindent  either (i) $z_s > z'_s$ or 
(ii) $z_s = z'_s$ and $z'_s(z_w - z'_w) - (z_\varphi - z'_\varphi ) \leq  0$. 
In either case, we have $F(z, \alpha (z-z'), X) > -\infty $ and 
\begin{eqnarray}\nonumber 
&&F_\lambda (z', r, \alpha (z-z'), Y) - F_\lambda (z, r, \alpha (z-z'), X) \\\nonumber 
&=& 
F(z', \alpha (z-z'), Y) - F(z, \alpha (z-z'), X)\\\nonumber 
&\leq &
\frac{1}{2}(\hat{\sigma }^2(z_s)X_{ss} - \hat{\sigma }^2(z'_s)Y_{ss}) + 
|\hat{b}(z_s) - \hat{b}(z'_s)|\alpha |z_s - z'_s|\\ \label{temp_proof}
&& + 
\alpha \sup _{\zeta \geq 0}\left\{ 
-(z_s - z'_s)^2g(\zeta ) + (z_s - z'_s)(z_w - z'_w)\zeta 
\right\} . 
\end{eqnarray}
Since (\ref {cond_XY}) implies 
\begin{eqnarray*}
\hat{\sigma }^2(z_s)X_{ss} - \hat{\sigma }^2(z'_s)Y_{ss} \leq 3\alpha (\hat{\sigma }(z_s) - \hat{\sigma }(z'_s))^2 
\end{eqnarray*}
and $\hat{\sigma }$ and $\hat{b}$ are both Lipschitz continuous and demonstrate linear growth, 
we have 
\begin{eqnarray*}
\frac{1}{2}(\hat{\sigma }^2(z_s)X_{ss} - \hat{\sigma }^2(z'_s)Y_{ss}) + 
|\hat{b}(z_s) - \hat{b}(z'_s)|\alpha |z_s - z'_s|
\leq  C_0\alpha |z_s - z'_s|^2
\end{eqnarray*}
for some $C_0 > 0$. 

Next we estimate the last term of the right side of (\ref {temp_proof}). 
If $z_s = z'_s$, it is obvious that this term is equal to zero, 
so we consider the case $z_s > z'_s$. 
Since $\liminf _{\zeta \rightarrow \infty }(h(\zeta )/\zeta ) > 0$, we see that 
there exist $\beta > 0$ and $\zeta _0 > 0$ such that 
$g(\zeta ) \geq \beta \zeta ^2$ for any $\zeta \geq \zeta _0$. 
Thus 
\begin{eqnarray*}
&&\sup _{\zeta \geq 0}\left\{ 
-(z_s - z'_s)^2g(\zeta ) + (z_s - z'_s)(z_w - z'_w)\zeta 
\right\} \\
&\leq &
(g(\zeta _0) + \zeta _0)|z - z'|^2 + 
\sup _{\zeta \geq 0}\left\{ 
-(z_s - z'_s)^2\beta \zeta ^2 + (z_s - z'_s)(z_w - z'_w)\zeta 
\right\} \\
&\leq & 
(g(\zeta _0) + \zeta _0)|z - z'|^2 + 
|z_w - z'_w|\left( \frac{z_w - z'_w}{2\beta } \vee 0\right)
\ \leq \ 
C_1|z-z'|^2
\end{eqnarray*}
for some $C_1 > 0$. 
Thus we obtain the assertion. 
\end{proof}

Now we present a proposition that includes the assertion of Theorem \ref {unique_th}. 

\begin{prop} \ \label{comparison_general}
Let $v$ $($resp., $v' )$ be functions such that 
\begin{eqnarray*}
|v(t, z)| + |v'(t, z)| \leq C(1 + z_w^2 + z_\varphi ^2 + z_s^2)^m, \ \ (t, z)\in [0, 1]\times \hat{D} 
\end{eqnarray*}
for some $C, m > 0$. 
Suppose that $v$ $($resp., $v')$ is a viscosity subsolution $($resp., supersolution$)$ of $(\ref {HJB2})$ on 
$(0, 1]\times \hat{D}$. 
Suppose further that $v$ and $v'$ satisfy $($\ref {bdd_cond}$)$. 
Then $v \leq v'$ on $[0, 1]\times \hat{D}$. 
\end{prop}

\begin{proof} 
Let $q(z) = (1 + z_w^2 + z_\varphi ^2 + z_s^2)^{m+1}$. 
By the similar arguments as in the proof of Proposition \ref {structure_prop}, we have 
\begin{eqnarray*}
|F(z, \mathcal {D}q(z), \mathcal {D}^2q(z))|\leq C_0q(z), \ \ z\in \hat{D}
\end{eqnarray*}
for some $C_0 > 0$. 
Let $\lambda > C_0$ and fix a value of $\varepsilon > 0$. 
We define $\bar{v}(t, z) = e^{-\lambda t}v(t, z) - \varepsilon q(z)$ and 
$\bar{v}'(t, z) = e^{-\lambda t}v'(t, z) + \varepsilon q(z)$. 
Then there exists an $R_\varepsilon > 0$ such that 
$\bar{v} < 0 < \bar{v}'$ holds on $[0, 1]\times \{ |z| \geq  R_\varepsilon  \}$. 
By a straightforward calculation, we see that $\bar{v}$ (resp. $\bar{v}'$) is a viscosity subsolution 
(resp. supersolution) of (\ref {HJB3}). Thus Proposition \ref {comparison_bdd} implies 
$\bar{v} \leq \bar{v}'$ on $[0, 1]\times \hat{D}$. 
Since $\varepsilon > 0$ was arbitrary, we obtain the assertion. 
\end{proof}

\subsection{Proof of Theorem \ref {converge}}\label{proof_of_converge}

We divide the proof of Theorem \ref {converge} into the following two propositions: 

\begin{prop} \ \label{converge_1} 
$\limsup _{n\rightarrow \infty }V^n_{[nt]}(w,\varphi ,s;u)\leq V_t(w,\varphi ,s;u)$. 
\end{prop}

\begin{prop} \ \label{converge_2} 
$\liminf _{n\rightarrow \infty }V^n_{[nt]}(w,\varphi ,s;u)\geq V_t(w,\varphi ,s;u)$. 
\end{prop}\vspace{0mm}\ \\
{\it Proof of Proposition \ref {converge_1}.} 
For brevity, we suppose $t=1$. 
For $u'\in \mathcal {C}$ and $(w',\varphi ',s')\allowbreak \in D$, let 
$\hat{\psi }_n(w',\varphi ',s' ; u')$ be an optimal strategy for the 
value function $V^n_1(w',\allowbreak \varphi ',s' ; u')$. 
By Proposition 7.33 in \cite {Bertsekas-Shreve}, we can take 
$\hat{\psi }_n(w',\varphi ',s' ; u')$ as a measurable function 
with respect to $(w', \varphi ', s')$. 
We define $(\psi ^n_l)^{n-1}_{l=0}\in 
\mathcal {A}^n_n(\varphi )$ and 
$(W^n_l,\varphi ^n_l, S^n_l)^n_{l=0}$ by $(W^n_0,\varphi ^n_0,S^n_0) = (w,\allowbreak \varphi ,s)$, 
$\psi ^n_l = \hat{\psi }_n(W^n_l, \varphi ^n_l, S^n_l ; V^n_{n-l-1}(\cdot ; u )) \wedge 
\varphi ^n_l$, 
(\ref {W_varphi})--(\ref {fluctuate_X}) inductively in $l$ and 
let $X^n_l = \log S^n_l$. 
Note that $(\psi ^n_l)_l$ is optimal, i.e., 
$V^n_n(w, \varphi , s ; u) = \E [u(W^n_n, \varphi ^n_n, S^n_n)]$. 
We also define a strategy $(\zeta _r)_{0\leq r\leq1}$ by 
$\zeta _r = n\psi ^n_{[nr]}$. Then 
$(\zeta _r)_r\in \mathcal {A}_1(\varphi )$. 
Let $(W_r,\varphi _r,X_r)_{0\leq r\leq 1} = \Xi ^X_1(w,\varphi ,s ; (\zeta _r)_r)$. \vspace{2mm}\\
{\it Step 1.} \ 
First we show that there is a constant $C^*>0$ and 
a sequence $(c^*_n)_{n\in \Bbb {N}}\subset (0,\infty )$ with $c^*_n/n\longrightarrow 0$ as 
$n\rightarrow \infty $ such that 
\begin{eqnarray*}
g_n(\psi ^n_l) \leq C^*\wedge (c^*_n\psi ^n_l), \ \ l=0,\ldots ,n-1. 
\end{eqnarray*}

If $h(\infty )<\infty $, the assertion is obvious. 
So we may assume $h(\infty ) = \infty $. 
Let $p_n(\psi ) = \psi e^{-g_n(\psi )}$ for $\psi \in [0,\Phi _0]$. 
This function implies the proceeds of liquidating $\psi $ shares of the security of the price $1$. 
The main intuition behind the following argument is that 
MI for a large sale is so large that 
larger sales result in smaller proceeds, that is, 
$p_n(\psi )$ is not increasing with respect to $\psi $, 
thus the optimal liquidation volumes at each time cannot become so large 
(for a typical example, when $g_n (\psi ) = n\alpha \psi ^2$ with $\alpha > 0$, 
the optimal volumes are smaller than $1/\sqrt{2\alpha n}$). 

We can easily see that $\frac{d}{d\psi }p_n(\psi ) = e^{-g_n(\psi )}(1-f_n(\psi ))$, 
where $f_n(\psi ) = \psi \frac{d}{d\psi }g_n(\psi )$. 
So the first-order condition $\frac{d}{d\psi }p_n(\psi ) = 0$ is equivalent to 
$f_n(\psi ) = 1$. 
Let
$A_n = \{ \psi \in (0,\Phi _0]\ ; \ f_n(\psi ) = 1\} $. 
By [A] and the assumption $h(\infty ) = \infty $, 
we see that $A_n$ is not empty and 
the function $p_n(\psi )$ has a maximum at one of the points in $A_n$ for sufficiently large $n$. 
We denote by $\psi ^*_n$ a point at which $p_n(\psi )$ has a maximum. 

We see that $p_n(\psi )\leq p_n(\psi ^*_n)$ 
for $\psi \in (\psi ^*_n, \Phi _0]$ and that 
Lemma \ref{cond_of_X3} implies 
that $Y(t ; r, x - g_n(\psi ))$ is non-increasing with respect to $\psi $. 
Moreover the function $u(w, \varphi ,s)$ is non-decreasing in $(w,\varphi ,s)$. 
Thus $\hat{\psi }_n(w,\varphi ,s ; u)\leq \psi ^*_n$ holds for large $n$. 
Then, by the definition of $\psi ^n_l$, we get 
\begin{eqnarray}\label{temp_optest}
\psi ^n_l\leq \psi ^*_n, \ \ l=0,\ldots ,n-1\ \ \mathrm {and} \ \ n > n_0
\end{eqnarray}
for some $n_0\in \Bbb {N}$. 
Moreover, [A] implies 
\begin{eqnarray}\label{psi^*}
n\psi ^*_n\longrightarrow \infty ,\ \ n\rightarrow \infty . 
\end{eqnarray}
Indeed, if (\ref {psi^*}) does not hold, there is a constant $M > 0$ and a subsequence $(n_k)_k\subset \Bbb {N}$ 
such that $n_k\psi ^*_{n_k} \leq M$. 
Then 
\begin{eqnarray*}
n_k = n_kf_{n_k}(\psi ^*_{n_k}) \leq n_k\psi ^*_{n_k}(h(n_k\psi ^*_{n_k}) + \varepsilon '_{n_k}) \leq M(h(M) + \varepsilon '_{n_k}) 
\end{eqnarray*}
for any $k$, where $\varepsilon '_n = \sup _{\psi }\left| \frac{dg_n}{d\psi }(\psi ) - h(n\psi )\right| $. 
This is a contradiction. 

Since $h(\zeta )$ is non-decreasing and $f_n(\psi ^*_n) = 1$, we have 
\begin{eqnarray}\label{temp_optest2}
g_n(\psi) \leq  
\Big(\frac{1}{\psi ^*_n} + 2\varepsilon '_n\Big)\psi , \ \ \psi \in [0,\psi ^*_n] 
\end{eqnarray}
for any $n\in \Bbb {N}$. 
By (\ref {temp_optest})--(\ref {temp_optest2}), 
we have the assertion by letting 
\begin{eqnarray*}
C^* = \max _{n \leq n_0}g_n(\Phi _0) + 1 + 2\Phi _0\sup _n\varepsilon '_n, \ \ 
c^*_n = \frac{1}{\psi ^*_n} + 2\varepsilon '_n. 
\end{eqnarray*}
{\it Step 2.} \ 
In this step we will show that 
\begin{eqnarray}\label{comparison_X}
\lim _{n\rightarrow \infty }\E [\max _{k=0,\ldots ,n}|X^n_k - X_{k/n}|^2] = 0. 
\end{eqnarray}

We define $\tilde{X}^n_r,\ r\in [0,1]$, by 
\begin{eqnarray}\label{def_of_tildeX}
\tilde{X}^n_r = Y\Big (r ; \frac{k}{n}, X^n_k-g_n(\psi ^n_k)\Big ),\ \ 
r\in \Big (\frac{k}{n}, \frac{k+1}{n}\Big ] 
\end{eqnarray}
and $\tilde{X}^n_0 = \log s$. 
Then we see that $\tilde{X}^n_{k/n} = X^n_k$ for each $k=0,\ldots , n$ and that 
$\tilde{X}^n_r$ satisfies 
\begin{eqnarray*}
\tilde{X}^n_r = \log s+\int ^r_0\sigma (\tilde{X}^n_v)dB_v + \int ^r_0b(\tilde{X}^n_v)dv - 
\sum ^{\lceil nr\rceil -1}_{k=0}g_n(\psi ^n_k), 
\end{eqnarray*}
where 
$\lceil \cdot \rceil $ is the ceiling function 
and $\sum ^{-1}_{k=0}g_n(\psi ^n_k) = 0$. 

Let $\Delta ^n_r = \E \Big [\max \Big \{ |\tilde{X}^n_{r'}-X_{r'}|^2\ ; \ 
r'=0,\frac{1}{n},\ldots ,\frac{[nr]}{n},r\Big \} \Big ]$. 
We have 
\begin{eqnarray}\label{temp_g_est}
\Big| \sum ^{\lceil nr'\rceil -1}_{k=0}g_n(\psi ^n_k) - \int ^{r'}_0g(\zeta _v)dv\Big| 
&\leq & 
\sum ^{n-1}_{k=0}\Big| g_n(\psi ^n_k)-\frac{1}{n}g(n\psi ^n_k)\Big| + 
d_n(r)g_n(\psi ^n_{[nr]}) \hspace{10mm}
\end{eqnarray}
for $r'=0,1/n,\ldots ,[nr]/n,r$, where 
$d_n(r) = \lceil nr\rceil - nr$. 
By Step 1, (\ref {temp_g_est}) and standard arguments using 
the Burkholder--Davis--Gundy inequality and the H\"older inequality, 
we can show that 
\begin{eqnarray*}
\Delta ^n_r \leq 
C_0\left \{ 
\gamma _n(r) + \E \left [ \int ^r_0|\tilde{X}_v-X_v|^2dv\right ] 
\right \} \leq 
C_0\left \{ \gamma _n(r) + \int ^r_0\Delta ^n_vdv\right \} 
\end{eqnarray*}
for some constant $C_0>0$, where 
$\gamma _n(r) = \Phi _0^2\varepsilon _n^2 + C^*c^*_nd_n(r)^2\E [\psi ^n_{[nr]}]$ and 
$\varepsilon _n$ is defined by (\ref {conv_g}). 
Now the generalised Gronwall inequality (see 
Lemma 10.5.1.3 in \cite {Dieudonne} 
for instance) implies 
\begin{eqnarray*}
\Delta ^n_r &\leq & C_0\gamma _n(r) + C_0^2
\int ^r_0\gamma _n(v)e^{C_0(r-v)}dv. 
\end{eqnarray*}
Since $0\leq d_n(v) \leq 1$ for $v\in [0,1]$ and $d_n(1)=0$, we have 
\begin{eqnarray}\label{temp_conv_1}
\E [\max _{k=0,\ldots ,n}|X^n_k-X_{k/n}|^2] = \Delta ^n_1 
\leq C_1\left \{ \Phi _0^2\varepsilon _n^2 + \frac{\Phi _0C^*c^*_n}{n}\right \} 
\end{eqnarray}
for some $C_1>0$. By (\ref {conv_g}) and the assertion of Step 1, 
the right side of (\ref {temp_conv_1}) tends to zero as $n\rightarrow \infty $. 
Then we have (\ref {comparison_X}). \vspace{2mm}\\
{\it Step 3.} \ 
Let $\tilde{W}^n_n = w+\sum ^{n-1}_{l=0}\int ^{(l+1)/n}_{l/n}
n\psi ^n_l\exp (X^n_l-(nr-l)g_n(\psi ^n_l))dr$. 
From (\ref {exp_log}), 
it follows that 
\begin{eqnarray}\label{est_In}
|\tilde{W}^n_n-W_1| 
\leq 
\Phi _0(\hat{Z}(s)+1)^2I_n, 
\end{eqnarray}
where 
\begin{eqnarray*}
I_n = \max _{l=0,\ldots ,n-1}\sup _{r\in [l/n,(l+1)/n]}|X^n_l - (nr-l)g_n(\psi ^n_l)-X_r|. 
\end{eqnarray*}
By Lemma \ref {cond_of_X1} and a straightforward calculation, we have 
\begin{eqnarray}\label{d_n}
\E [I_n^2]^{1/2}
\leq 
\frac{C_0}{n^{1/4}} + \E [\max _{k=0,\ldots ,n}|X^n_k - X_{k/n}|^2]^{1/2} + \Phi_0\varepsilon _n
\end{eqnarray}
for some $C_0>0$. 
From (\ref {comparison_X}), (\ref {est_In}), and Lemma \ref {cond_of_X2}, 
we get the convergence $\E [|\tilde{W}^n_n-W_1|] \longrightarrow 0$ as $n\rightarrow \infty $. 
On the other hand, (\ref {exp_log}), (\ref {comparison_X}), and Lemma \ref {cond_of_X2} yield 
$\E [|S^n_n-\exp (X_1)|]\allowbreak \longrightarrow 0$. 
Therefore, we can apply Lemma \ref {conti_u} to obtain 
\begin{eqnarray}\label{temp_step3}
\lim _{n\rightarrow \infty }
\big| \E [u(\tilde{W}^n_n, \varphi ^n_n, S^n_n)] - 
\E [u(W_1, \varphi _1, \exp (X_1))]\big| = 0. 
\end{eqnarray}
Since $(\psi ^n_l)_l$ is optimal, 
$u$ is non-decreasing in $w$, and $\tilde{W}^n_n\geq W^n_n$, we have 
\begin{eqnarray}\label{temp_step3_2}
V^n_n(w,\varphi ,s ; u)-V_1(w,\varphi ,s ; u)
\leq 
\E [u(\tilde{W}^n_n,\varphi ^n_n,S^n_n)] - \E [u(W_1, \varphi _1, \exp (X_1))]. 
\hspace{5mm}
\end{eqnarray}
Now the assertion of Proposition \ref {converge_1} 
is given by (\ref {temp_step3}) and (\ref {temp_step3_2}). \qed \\\\
{\it Proof of Proposition \ref {converge_2}.} 
Again we suppose $t=1$. \enlargethispage*{2\baselineskip }
Take any $(\zeta _r)_{0\leq r\leq 1}\in \mathcal {A}_1(\varphi )$ and let 
$\psi ^n_l = \int ^{l/n}_{((l-1)/n)\vee 0}\zeta _rdr$\vspace{1mm}, where 
$a\vee b = \max \{ a,b \}$. 
Then we have $(\psi ^n_l)_l\in \mathcal {A}^n_n(\varphi )$. 
Let 
$(W_r,\varphi _r,X_r)_{0\leq r\leq 1} = 
\Xi ^X_1(w,\varphi ,s ; (\zeta _r)_r)$ and 
$(W^n_l, \varphi ^n_l, S^n_l)^n_{l=0} = 
\Xi ^n_n(w,\allowbreak \varphi ,s ; \allowbreak (\psi ^n_l)_l)$. 
Put $X^n_l = \log S^n_l$. 
By arguments similar to those used in the proof of 
Proposition \ref {converge_1} and Lebesgue's differentiation theorem, 
we get 
$\E [\max _{k=0,\ldots ,n}|X^n_k \allowbreak - X_{k/n}|^2]\longrightarrow 0$ 
as $n\rightarrow \infty $, which also implies 
$\E [|S^n_n - \exp (X_1)|] \longrightarrow  0$. 

Next, let $\hat{W}^n_1 = w+\sum ^{n-1}_{l=0}\psi ^n_ln\int ^{(l+1)/n}_{l/n}\exp (X_r)dr$. 
By a straightforward calculation, we have 
$\E [|\hat{W}^n_1-W_1|]
\leq 
C_0\Big\{ \Phi _0\tilde{I}_n + 
K_n + n^{-1} \Big\} $ 
for some $C_0>0$ depending only on $b, \sigma ,(\zeta _r)_r$ and $s$, where 
\begin{eqnarray*}
&&\tilde{I}_n = \E [\sup _{v\in [0,1-1/n]}|X_{v+1/n}-X_v|^2]^{1/2}, \\
&&K_n = \Big( \int ^1_0\E [|H_n(r)|^2]dr\Big) ^{1/2}, \ \ 
H_n(r) = n\int ^{([nr]+1)/n}_{[nr]/n}\zeta _vdv-\zeta _r. 
\end{eqnarray*}
Lebesgue's differentiation theorem and the dominated convergence theorem imply 
$K_n\longrightarrow 0$. 
From $\sup _{r, \omega }\zeta _r(\omega ) < \infty $ and 
Lemma \ref {cond_of_X1}, we can easily show that 
$\tilde{I}_n \longrightarrow 0$. 
Then we obtain $\E [|\hat{W}^n_1-W_1|] \longrightarrow  0$. 
On the other hand, a similar calculation to Step 2 of the proof of Proposition \ref {converge_1} implies 
$\E [|W^n_n-\hat{W}^n_1|]\longrightarrow 0$. 
Thus $\E [|W^n_n-W^n_1|] \longrightarrow 0$ converges. 
Then we can apply Lemma \ref {conti_u} and we get 
\begin{eqnarray*}
\E [u(W_1,\varphi _1,\exp (X_1))] = 
\lim _{n\rightarrow \infty }\E [u(W^n_n,\varphi ^n_n,S^n_n)] \leq 
\liminf _{n\rightarrow \infty }V^n_n(w,\varphi ,s ; u). 
\end{eqnarray*} 
Since $(\zeta _r)_r\in \mathcal {A}_1(\varphi )$ is arbitrary, 
we obtain the assertion. \qed \

\subsection{Proof of Proposition \ref {prop_deterministic}}\label{proof_deterministic}

First we introduce the following lemma: 

\begin{lem} \ \label {deterministic_discrete} 
Under \rm {[A]} and the assumptions of Section \ref {sec_eg}, it holds that 
\begin{eqnarray}\label{temp_Vnk_f}
V^n_k(w, \varphi , s ; u_{\mathrm {RN}}) = w + sf^n_k(\varphi ),
\end{eqnarray} 
where $f^n_k$ is defined by (\ref {f_n_k1})--(\ref {f_n_k3}). 
\end{lem}

\begin{proof} 
This lemma is proved by mathematical induction. 
First, the assertion is obvious when $k = 0$. 
Next, we assume that $V^n_k(w, \varphi , s ; u_{\mathrm {RN}}) = w + sf^n_k(\varphi )$ for some $k$. 
Then the standard arguments of the Bellman equation (\cite {Bertsekas-Shreve}) imply 
\begin{eqnarray}\nonumber 
V^n_{k+1}(w, \varphi , s ; u_{\mathrm {RN}}) &=& 
\sup _{\psi \in [0, \varphi ]}
\E [V^n_k(w + \psi se^{-g_n(\psi )}, \varphi - \psi , se^{-\mu /n + \sigma B_{1/n} - g_n(\psi )})]\\
&=& w + s\sup _{\psi \in [0, \varphi ]}\{ \psi e^{-g_n(\psi )} + e^{-\tilde{\mu }/n-g_n(\psi )}f^n_k(\varphi - \psi )\} . 
\label{temp_discrete_0}
\end{eqnarray}
Now take any $\psi \in [0, \varphi ]$. 
The arguments in the beginning part of the proof of Proposition \ref {converge_1} tell us that 
$f^n_k(\varphi - \psi )$ can be written as 
\begin{eqnarray*}
f^n_k(\varphi - \psi ) = 
\sum ^{k-1}_{l = 0}\hat{\psi }^n_l\exp \left( -\tilde{\mu }\times \frac{l}{n} - \sum ^{l}_{m=0}g_n(\hat{\psi }^n_m)\right) \\
\end{eqnarray*}
for some $(\hat{\psi }^n_l)^{k-1}_{l=0}\in \mathcal {A}^{n, \mathrm {det}}_k(\varphi - \psi )$. 
Then we have 
\begin{eqnarray*}&&
\psi e^{-g_n(\psi )} + e^{-\tilde{\mu }/n - g_n(\psi )}f^n_k(\varphi - \psi )
\leq 
f^n_{k+1}(\varphi )
\end{eqnarray*}
because of $\psi + \sum ^{k-1}_{l=0}\hat{\psi }^n_l\leq \varphi $ 
(note that $(\psi , \hat{\psi }^n_0, \ldots , \hat{\psi }^n_{k-1})\in \mathcal {A}^{n,\mathrm {det}}_{k+1}(\varphi )$). 
By the above inequality and (\ref {temp_discrete_0}), we get 
$V^n_{k+1}(w, \varphi , s ; u_{\mathrm {RN}}) \leq w + sf^n_{k+1}(\varphi )$. 
The opposite inequality 
is easily obtained 
by the relation $\mathcal {A}^{n, \mathrm {det}}_k(\varphi )\subset \mathcal {A}^{n}_k(\varphi )$. 
Then we have $V^n_{k+1}(w, \varphi , s ; \allowbreak u_{\mathrm {RN}}) = w + sf^n_{k+1}(\varphi )$, and 
the proof is completed. 
\end{proof}

Now we prove Proposition \ref {prop_deterministic}. 
Applying Theorem \ref {converge} and Lemma \ref {deterministic_discrete}, 
we obtain 
\begin{eqnarray*}
V^n_t(w, \varphi , s ; u_\mathrm {RN}) = \lim _{n\rightarrow \infty }
V^n_{[nt]}(w, \varphi , s ; u_\mathrm {RN}) = w + s\lim _{n\rightarrow \infty }f^n_{[nt]}(\varphi ) = 
w + sf(t, \varphi ), 
\end{eqnarray*}
which imply the assertion.  \qed

\subsection{Proof of Theorem \ref {th_LL}}\label{sec_proof_LL}
Assertion (i) is directly obtained by ($9'$)--($12'$) in \cite {Lions-Lasry}. 
Now we prove (\ref {eq_LL}) under the assumptions of assertion (ii). 
We can show the inequality 
$V^\mathrm {SO}_t(w, \varphi , s ; U) \geq  U\left( w + \frac{1 - e^{-\alpha \varphi }}{\alpha }s\right) $ 
by considering strategy (\ref {almost_block}) and letting $\delta \downarrow 0$. 
To see the opposite inequality, it suffices to show that 
$\overline{V}^\varphi _t(\bar{w}, \bar{s}) \leq U(\bar{w})$ 
for any $\bar{w}$ and $\bar{s}$. 
But this is easily obtained because we have the inequality 
\begin{eqnarray*}
\E [U(\bar{W}_t)] \leq  U(\E [\bar{W}_t]) = 
U\left( \bar{w} + \int ^t_0
\E \left[ \frac{1 - e^{-\alpha \bar{\varphi }_r}}{\alpha }\hat{b}(\bar{S}_re^{\alpha \bar{\varphi }_r})\right] dr\right) 
\leq  U(\bar{w}) 
\end{eqnarray*}
for each $(\bar{\varphi }_r)_r\in \overline {\mathcal {A}}_t(\varphi )$; 
This can be proved from the observations that 
$U$ is concave and non-decreasing and that $\hat{b}$ is non-positive. \qed 

\ \\
{\bfseries Acknowledgements}: \ 
The author would like to thank the anonymous referees 
for their valuable comments and suggestions. 
The author also thanks Prof. S. Kusuoka from the Graduate School of Mathematical Sciences, 
The University of Tokyo, 
Prof. J. Sekine from the Graduate School of Engineering Science, Osaka University and 
Prof. K. Ishitani from 
Department of Mathematics, Meijo University 
for their helpful advice and discussions. 

This article is the preprint version of the article 
``An optimal execution problem with market impact'' 
published in {\it Finance and Stochastics}, 
DOI: 10.1007/s00780-014-0232-0. 
The final publication is available at link.springer.com: \\
\url{http://link.springer.com/article/10.1007/s00780-014-0232-0}

\end{document}